\documentclass[11pt]{article} 

\usepackage[margin=1in]{geometry} 
\usepackage[utf8]{inputenc}
\usepackage[T1]{fontenc}
\usepackage{mathpazo} 
\usepackage{enumerate}
\usepackage{amsmath,amssymb}
\usepackage{authblk,comment,xspace}
\usepackage[normalem]{ulem} 
\usepackage{booktabs,colortbl,multicol,multirow,subcaption} 
\usepackage[ruled]{algorithm2e} 
\usepackage{enumitem}

\SetAlFnt{\small}
\SetAlCapFnt{\small}
\SetAlCapNameFnt{\small}
\SetAlCapHSkip{0pt}
\IncMargin{-\parindent}

\usepackage{crossreftools}
\usepackage{amsmath,amsfonts,amssymb,amsthm,bbm,mathtools,pifont,thm-restate}

\usepackage[dvipsnames]{xcolor} 
\definecolor{ptblue}{RGB}{15,76,129} 
\definecolor{ptemerald}{HTML}{009473} 
\definecolor{ptgray}{HTML}{939597} 

\usepackage{tikz}
\usetikzlibrary{positioning,shapes,arrows,graphs,quotes,automata,backgrounds}
\usepackage{mathtools}
\usepackage[square]{natbib} 
\usepackage{hyperref}
\hypersetup{
	linktocpage,
	hidelinks=true,
	citecolor=cobalt, 
	urlcolor=ptblue, 
	linkcolor=cobalt, 
}
\definecolor{cobalt}{rgb}{0.0, 0.28, 0.67}
\usepackage{cleveref} 

\theoremstyle{plain}
\newtheorem{theorem}{Theorem}[section]
\newtheorem{corollary}[theorem]{Corollary}	
\newtheorem{proposition}[theorem]{Proposition}
\newtheorem{lemma}[theorem]{Lemma}

\theoremstyle{definition}
\newtheorem{definition}[theorem]{Definition}
\newtheorem*{definition*}{Definition}

\theoremstyle{remark}
\newtheorem{example}[theorem]{Example}

\allowdisplaybreaks

\newtheorem*{theorem*}{Theorem}

\theoremstyle{remark}
\newtheorem{remark}{\upshape\bfseries Remark}

\DeclareMathOperator*{\argmax}{arg\,max}

\makeatletter 
\newcommand{\indicator}[1]{\mathbbm{1}_{\left\{#1\right\}}\xspace}

\newcommand{\committeeSize}{k}
\newcommand{\approvalBallotOfAgent}[1]{A_{#1}} 

\newcommand{\RUT}{\text{RUT}\xspace}

\newcommand{\GRP}{\text{GRP}\xspace}
\newcommand{\instance}{\mathcal{I}\xspace}
\newcommand{\network}{\mathcal{N}\xspace}
\newcommand{\dnetwork}{\mathcal{D}\xspace}
\newcommand{\wnetwork}{\mathcal{W}\xspace}
\newcommand{\GCUT}{Generalized CUT\xspace}

\newcommand{\tildG}{\widetilde{g}}
\newcommand{\cost}{\texttt{cost}}
\newcommand{\capacity}{\texttt{cap}}

\title{Maximum Flow is Fair: A Network Flow Approach to Committee Voting}
\author{Mashbat Suzuki \qquad Jeremy Vollen\thanks{Corresponding author.}
\medskip\\
UNSW Sydney \medskip\\
\nolinkurl{{mashbat.suzuki, j.vollen}@unsw.edu.au}}
\date{}

\usepackage{abstract}

\begin{document}
\maketitle

\begin{abstract}
	In the committee voting setting, a subset of $k$ alternatives is selected based on the preferences of voters.
	In this paper, our goal is to efficiently compute \emph{ex-ante} fair probability distributions over committees.
	We introduce a new axiom called \emph{group resource proportionality}, which strengthens other fairness notions in the literature. 
	We characterize our fairness axiom by a correspondence with max flows on a network formulation of committee voting.
	Using the connection to flow networks revealed by this characterization, we introduce two voting rules which achieve fairness in conjunction with other desiderata.
	The first rule --- the \emph{redistributive utilitarian rule} --- satisfies ex-ante efficiency in addition to our fairness axiom.
	The second rule --- \emph{Generalized CUT} --- reduces instances of our problem to instances of the minimum-cost maximum-flow problem.
	We show that \GCUT maximizes social welfare subject to our fairness axiom and additionally satisfies an incentive compatibility property known as \emph{excludable strategyproofness}.
	Lastly, we show our fairness property can be obtained in tandem with strong \emph{ex-post} fairness properties -- an approach known as \emph{best-of-both-worlds} fairness.
	We strengthen existing best-or-both-worlds fairness results in committee voting and resolve an open question posed by \citet{ALS+23}.
\end{abstract}

\section{Introduction}
\label{sec:intro}

In the \emph{committee voting} problem, we are tasked with selecting a subset of $k$ alternatives (or \emph{candidates}) given the preferences of a population of voters.
The problem models a rather natural and practical scenario and has garnered notable attention in recent research~\citep{BCE+16a,ABES19,Endr17a}.
In this work, we contribute to the growing body of work pursuing fair committee voting outcomes when voters report preferences in the form of a subset of candidates of which they \emph{approve} (see \citet{LaSk23a} for extensive coverage of this topic).
One clear advantage of approval ballots is due to their simplicity compared to ballots such as rankings and cardinal reports. 
Furthermore, approval ballots are clearly the most natural method of preference elicitation under our assumptions about voters' preference structure: we assume voters' preferences are \emph{dichotomous} with respect to candidates and derive utility equal to the number of approved candidates on the selected committee.

While the assumption of dichotomous preferences allows for many theoretical results that encounter impossibilities in more general settings, there remain significant obstacles in obtaining committees which are fair to all voters.
Most fundamentally, no deterministic rule satisfies a very basic fairness notion known as \emph{positive share}, which guarantees that every voter receives some non-zero representation on the selected committee.
A natural approach to circumvent this and other impossibilities is to employ randomization by computing a probability distribution (or \emph{lottery}) over committees which achieves desirable properties in an \emph{ex-ante} sense.
This approach has significant precedent in the (single-winner) voting setting \citep{Gibbard77a}, and in particular, voting under dichotomous preferences \citep{BMS05a,ABM19a,BBP+21a}. 

Recently, this approach has been extended to the committee voting setting with dichotomous preferences \citep{CJM+20a,ALS+23}.
Fairness axioms in committee voting often aim for some semblance of proportional representation, the idea that every group of voters should control a fraction of the committee proportional to their size. 
The most pervasive such axiom is that of \emph{core}, initially conceptualized and investigated by \citet{Dro81} and \citet{Lind58}. 
While existence of the core is not known in an ex-post sense, its fractional analog known as \emph{fractional core} has been shown to exist.
However, lotteries satisfying fractional core may behave quite counter-intuitively with respect to proportional representation, as illustrated by the following  example.
\begin{example}
	\label{ex:core_bad}
	Suppose there are four voters, denoted $\{1,2,3,4\}$,  three candidates, denoted $\{a,b,c\}$, and the desired committee size is $k=2$.
	Suppose voter $1$ approves of candidate $\{a\}$, voters $2$ and $3$ approve of candidates $\{a,b\}$, and voter $4$ approves of candidate $\{c\}$.
	Let $\Delta$ be the lottery which selects the committee $\{a,b\}$ with probability $\frac{1}{3}$ and committee $\{a,c\}$ with probability $\frac{2}3$. 
	Consider the group of voters $S=\{1,2,3\}$.
	Since they constitute $\frac{3}4$ of the population of voters, it seems natural that they should control a $\frac{3}4$ fraction of the committee, in expectation, in any lottery providing proportional representation \emph{ex-ante}.
	That is, one would expect that the number of selected candidates that some voter in $S$ approves, in expectation, is at least $\frac{3}2$.
	Instead, there are only $1\cdot \frac{2}{3} + 2\cdot \frac{1}3 = \frac{4}3$ such candidates, in expectation. 
	
	Thus, the voter group $S$ could justifiably complain that they are underrepresented by the lottery $\Delta$.
	One can check that $\Delta$, in fact, satisfies fractional core (see \Cref{def:core}).
	Furthermore, we can see that the lottery which selects the committees $\{a,b\}$ and $\{a,c\}$ with equal probability does not suffer from this under-representation issue for any subset of voters.
	In summary, the requirement of fractional core is not sufficient to guarantee the selection of proportionally representative outcomes, even when such outcomes are guaranteed to exist. 
\end{example}

Besides the behavior shown in \Cref{ex:core_bad}, fractional core presents additional challenges.
There is no known polynomial time algorithm for computing fractional core \citep{MSW22}.
Furthermore, while maximizing Nash welfare computes a fractional core outcome in the single-winner setting, this does not extend to committee voting (see \Cref{ex:nash}).
This highlights the structural differences between the two settings.
Indeed, \citet{ALS+23} demonstrate the technical subtleties that arise when extending axioms and algorithms from probabilistic single-winner voting to probabilistic committee voting.
They extend the fair share axioms from the single winner setting \citep{Dudd15a,BMS05a} to the committee setting, showing that two alternative interpretations of those axioms result in distinct hierarchies, the strongest properties of which are \emph{group fair share} and \emph{strong unanimous fair share}.

\subsection*{Contributions}
Our first contribution is the \emph{group resource proportionality} (\GRP) axiom (\Cref{sec:grp}), which unifies both axiom hierarchies from \citet{ALS+23} (by strengthening both group fair share and strong unanimous fair share), and is achievable in polynomial time.
Furthermore, \GRP is impervious to the type of under-representation issue allowed by fractional core in \Cref{ex:core_bad}.
At a high level, \GRP lower bounds, for each group of voters $S$, the number of selected candidates which represent some voter in $S$, in expectation. 
The lower bound depends on the proportional size and the preference structure of the voter group. 
We provide a characterization of \GRP (\Cref{thm:gr_characterization}) using a network flow formulation of our problem, in which voters control a proportional fraction of the committee size (or ``budget''), and can only flow into candidates they approve. 
Specifically, we show that a lottery satisfies our axiom if and only if it corresponds to a solution to the max flow problem on this network.
As a result of this characterization, we get that \GRP is not only polynomial time computable, but can also be checked in polynomial time.

We then exploit the connection drawn between probabilistic committee voting and network flows in order to devise fair voting rules which achieve additional desiderata.
In \Cref{sec:efficiency}, we introduce a voting rule which obtains \GRP and ex-ante (Pareto) efficiency and can be computed in polynomial time. 
Our algorithm -- the \emph{redistributive utilitarian rule} -- leverages our characterization and carefully redistributes voters' remaining budgets to avoid any voter exhausting their budget prematurely.

In the single-winner setting, \citet{BBP+21a} showed that ex-ante efficiency and strategyproofness cannot be achieved jointly with positive share.
In that setting, the conditional utilitarian rule (CUT) offers a compromise: strategyproofness and efficiency subject to group fair share.
In \Cref{sec:gcut}, we give an impossibility demonstrating that no rule satisfying these properties exists in probabilistic committee voting.
In light of this impossibility, we introduce \emph{\GCUT}, a voting rule which maximizes social welfare subject to \GRP and satisfies the weaker notion of \emph{excludable strategyproofness}. 
To achieve these properties, our algorithm frames instances of probabilistic committee voting using a minimum-cost maximum-flow formulation.

Lastly, in \Cref{sec:bbw}, we show ex-ante \GRP can be obtained in tandem with strong ex-post fairness properties, such as FJR and EJR+, contributing to the literature on ``best of both worlds fairness'' \citep{ALS+23,ALSV24}. 
To do so, we identify a condition on an integral committee $W$ which is sufficient to guarantee polynomial computation of an ex-ante \GRP lottery over committees which contain $W$ (\Cref{thm:gr_wp_bbw}).
As corollaries of this statement, we get the strongest known best-of-both-worlds results in the committee voting setting, and resolve an open question posed by \citet{ALS+23}.

\subsection*{Related Work}

Approval-based committee voting is a fundamental voting model with broad applications. Multiwinner voting is readily applicable in settings where a fixed committee size must be chosen, such as in parliamentary elections , and the selection of a board members. Due to its generality, in recent years, it has found applications in recommender systems \citep{CGN19,GF22,SLB17}, the design of Q \& A platforms \citep{IB21}, blockchain protocols \cite{BCCH20}, and global optimization \cite{FSS17}.

As the existence of core is a major open problem in approval-based committee voting \cite{LaSk23a}, there has been a significant body of work focused on developing algorithms that achieve approximate notions of core. 
\citet{PeSk20a} showed that Proportional Approval Voting gives a committee that is a factor-two multiplicative approximation to core. 
\citet{MSW22} showed that a constant factor approximation to core is achievable in committee voting, even with monotone submodular utility functions. 
Another way to relax the notion of core is to restrict the deviating groups. 
This was the approach taken by \citet{ABC17}, who initiated the study of representation for cohesive groups, spawning a large body of work exploring various \emph{justified representation} axioms \cite{AEH+18a,EFI+22a,BFJL23,BP23}.  

Our paper achieves fairness in committee voting through a probabilistic lens. 
The study of lotteries over outcomes for a fixed preference profile is referred to as \textit{probabilistic voting}. 
In single-winner probabilistic voting, \citet{BMS05a} defined various ex-ante fairness notions and rules, which have been further explored in later works \citep{ABM19a,BBP+21a,MPS20a,Dudd15a}.
In the  committee voting setting, \citet{MSW22} used the existence of Lindahl equilibria along with various rounding schemes to prove existence of constant factor approximate core committees. 
\citet{CJM+20a} introduced the concept of a \textit{stable lottery}, which asserts that for any committee $W$ of size $\alpha$, the expected number of voters who prefer $W$ to the committee sampled from the lottery does not exceed $\alpha \frac{n}{k}$ for each $1\leq \alpha\leq k$. 
They use the probabilistic method to show the existence of stable lotteries. 
Unlike our fairness notion, no efficient algorithm is known to give stable lotteries under dichotomous preferences.

\citet{BLS24a} introduced the \emph{cake sharing} model, which is more general than our own.
In cake sharing, agents have piecewise uniform utilities over some divisible resource (or ``cake''), represented by the unit interval. 
The problem asks how to select a subset of the cake subject to a length constraint.
The authors show that the Leximin mechanism is truthful when allowing for blocking agents from pieces of cake that they claim not to approve.
This is equivalent to showing that Leximin is excludable strategyproof, a property we will study in \Cref{sec:gcut}. 
However, we obtain this property in conjunction with our fairness property, which is much stronger than any fairness property achieved by the Leximin mechanism.\footnote{Even in the single-winner special case of our model, Leximin fails to satisfy unanimous fair share, which is weaker than each of the fairness definitions we discuss in this work \citep{ABM19a}.
The authors also introduce fairness concepts which adapt notions of justified representation and average fair share \citep{ABM19a}.
In our setting, their properties are incomparable with our own.
Critically, their properties give guarantees only to groups of agents with sufficiently cohesive preferences, whereas we strive to give guarantees to all subsets of agents, as is done in the core.
}

When studying lotteries that are fair in expectation, it is natural to ask whether such a lottery can be supported over desirable outcomes. 
This is referred to as "best-of-both-worlds fairness," a term coined by \citet{FSV20b}. They show that ex-ante envy-freeness and ex-post near envy-freeness is achievable in the resource allocation setting. 
Subsequent works \cite{FMNP24,HSV23a} have established similar guarantees are possible in various resource allocation problems. 
In approval-based committee voting, \citet{ALS+23} showed that ex-ante GFS and ex-post EJR can be achieved simultaneously. 
The best-of-both-worlds perspective on voting was further explored by \citet{ALSV24} in the setting of participatory budgeting.

Our approach to exploring ex-ante fairness makes use of network flows, where flow conservation naturally aids in the feasible exchange of probability weights among voters. A similar flow-based approach was used by \citet{Vazi07} for the computation of market equilibria in a private goods economy and when studying competitive equilibria in trade networks \citep{CEV16}.

\section{Preliminaries}
\label{sec:prelim}
For any positive integer~$t \in \mathbb{N}$, we write $[t] \coloneqq \{1, 2, \dots, t\}$.
Denote $C$ as the set of $m$ \emph{candidates} and $N = [n]$ as the set of voters.
We assume that the voters have \emph{dichotomous} preferences over candidates. 
That is, each voter's preference is represented by a set of approved candidates, called an \emph{approval set}. 
For each voter $i \in N$, denote $\approvalBallotOfAgent{i} \subseteq C$ as the approval set of voter $i$. 
An instance~$\instance$ of approval-based committee voting is given by a set of candidates~$C$, an \emph{approval profile}~$\mathcal{A} = (\approvalBallotOfAgent{1}, \approvalBallotOfAgent{2}, \dots, \approvalBallotOfAgent{n})$, and a positive integer $k$ with $\committeeSize \leq m$.

\medskip
\noindent\textbf{Integral and Fractional Committees}.
An \emph{integral committee}~$W$ (or simply committee, when context is clear) is a subset of~$C$ of size~$\committeeSize$.
A \emph{fractional committee} is specified by an $m$-dimensional vector~$\vec{p} = (p_c)_{c \in C}\in [0,1]^m$ with $\sum_{c \in C} p_c = \committeeSize$.
For notational convenience, we use $\vec{1}_W \in \{0, 1\}^m$ to denote the vector representation of an integral committee~$W$, i.e. the fractional committee with the $j^{\text{th}}$ component equal to~$1$ if and only if~$j \in W$.

\medskip

\noindent \textbf{Lotteries.}
A \textit{lottery} is a probability distribution over integral committees. A lottery is specified by a set of~$s \in \mathbb{N}$ tuples $\{(\lambda_j, W_j)\}_{j \in [s]}$ with $\sum_j \lambda_j = 1$, where for each~$j \in [s]$, the integral committee $W_j \subseteq C$ is selected with probability $\lambda_j \in [0, 1]$. 
It can be seen that every lottery $\{(\lambda_j, W_j)\}_{j \in [s]}$ has a unique corresponding fractional committee given by $\vec{p} = \sum_{j \in [s]} \lambda_j \vec{1}_{W_j}$. Here $p_c$ can be interpreted as the marginal probability of candidate $c$ being selected in a committee drawn from the lottery. We say lottery $\{(\lambda_j, W_j)\}_{j \in [s]}$ is an \emph{implementation} of a fractional committee~$\vec{p}$ if $\vec{p} = \sum_{j \in [s]} \lambda_j \vec{1}_{W_j}$.

As is common in approval-based committee voting, we assume voters' preferences over committees are determined by the number of approved candidates on the committee, i.e. $u_i(W) = \vert A_i \cap W \vert.$
Extending to fractional committees, we define the utility of voter $i$ from fractional committee $\vec{p}$ as $u_i(\vec{p}) \coloneqq \sum_{c \in \approvalBallotOfAgent{i}} p_c$.
We point out that voter $i$ derives utility from a fractional committee $\vec{p}$ equivalent to his expected utility from any lottery implementing $\vec{p}$, i.e. $u_i( \vec{p} ) = \mathbb{E}_{W\sim \Delta}[u_i(W)]$for any implementation $\Delta$ of $\vec{p}$.

\subsection*{Properties of Voting Rules}

A \textit{(voting) rule} $F$ is a function that maps each approval profile $\mathcal{A}$ and committee size~$k$ to a feasible fractional committee $F(\mathcal{A},k)\in [0,1]^m$.  
We will use $F_c(\mathcal{A},k)$ to denote the amount of candidate $c\in C$ selected in the fractional committee selected by $F$.
That is, if $F(\mathcal{A},k) = \vec{p}$, then $F_c(\mathcal{A},k) = p_c$.
In this paper, we will introduce voting rules which achieve fairness in conjunction with other desiderata, the first of which is \emph{(Pareto) efficiency}.

\begin{definition}[Efficiency]
	A fractional committee $\vec{p}$ is \emph{efficient}  if there is no alternative fractional committee $\vec{q}$ such that
	$u_i(\vec{q})\geq u_i(\vec{p})$ for all $i\in N$ and this inequality is strict for at least one voter.
	A rule is efficient if it always returns an efficient fractional committee.
\end{definition}

We will also consider strategic properties of rules.
A rule satisfies \emph{strategyproofness} if no voter ever has an incentive to misreport her true preferences to the rule.
In \Cref{sec:gcut}, we show a strong impossibility with strategyproofness, and thus also consider a natural relaxation known as \emph{excludable strategyproofness} \citep{ABM19a}.
Excludable strategyproofness guarantees that no voter has an incentive to misreport if she cannot benefit from candidates which she claims not to approve.

\begin{definition}[Strategyproofness]
	A rule $F$ is \textit{strategyproof} if for all $i\in N$, all $k\leq m$, and all approval profiles of the form $\mathcal{A}=(A'_1,\ldots,A_i,\ldots,A'_n)$ and $\mathcal{A}'=(A'_1,\ldots,A'_i,\ldots,A'_n)$,  we have $u_i(F(\mathcal{A},k))\geq u_i(F(\mathcal{A}',k))$.
\end{definition}

\begin{definition}[Excludable Strategyproofness]
	A rule $F$ is \textit{excludable strategyproof} if for all $i\in N$, all $k\leq m$, and all approval profiles of the form $\mathcal{A}=(A'_1,\ldots,A_i,\ldots,A'_n)$ and $\mathcal{A}'=(A'_1,\ldots,A'_i,\ldots,A'_n)$, we have $u_i(F(\mathcal{A},k))\geq \sum_{c\in A_i\cap A'_i} F_c(\mathcal{A'},k)$.
\end{definition}

We point out that, while we define voting rules and most of the properties in this paper in terms of fractional committees, each of our properties and results have natural analogs for lotteries. 
Specifically, we say a lottery satisfies a property ex-ante  if and only if the lottery implements a fractional committee which satisfies the fractional property.
As an example, if fractional committee $\vec{p}$ satisfies strategyproofness (as above), then any algorithm that outputs an integral committee sampled from a lottery which implements $\vec{p}$ will in fact satisfy \emph{ex-ante} strategyproofness.
Although every lottery corresponds to a unique fractional committee, the converse is not true. 
However, computation of an implementation of a fractional committee can be achieved through rounding schemes commonly studied in combinatorial optimization \cite{S01,GKPS06,CVZ10}. 
In particular, the rounding scheme of \cite{ALM+19a} shows that every fractional committee of size $k$ can be implemented with a lottery over integral committees of size $k$ in polynomial time.
This means that the algorithmic results of \Cref{sec:efficiency} and \Cref{sec:gcut} also carry over to lotteries.
In summary, while we will often restrict our treatment to fractional committees, this is purely for the sake of notational clarity.

\subsection*{Flow Network Formulation}
We will often represent committee voting instances using flow networks. 
At a high level, the network representation of an instance: (i) connects the source to a node for each voter, (ii) connects each voter to candidates of which they approve, and (iii) connects those candidates to the sink. 
Refer to Figure~\ref{fig:1} for an illustration of a network representation of an instance.
\begin{figure}
	\begin{center}
		\includegraphics[scale=0.22]{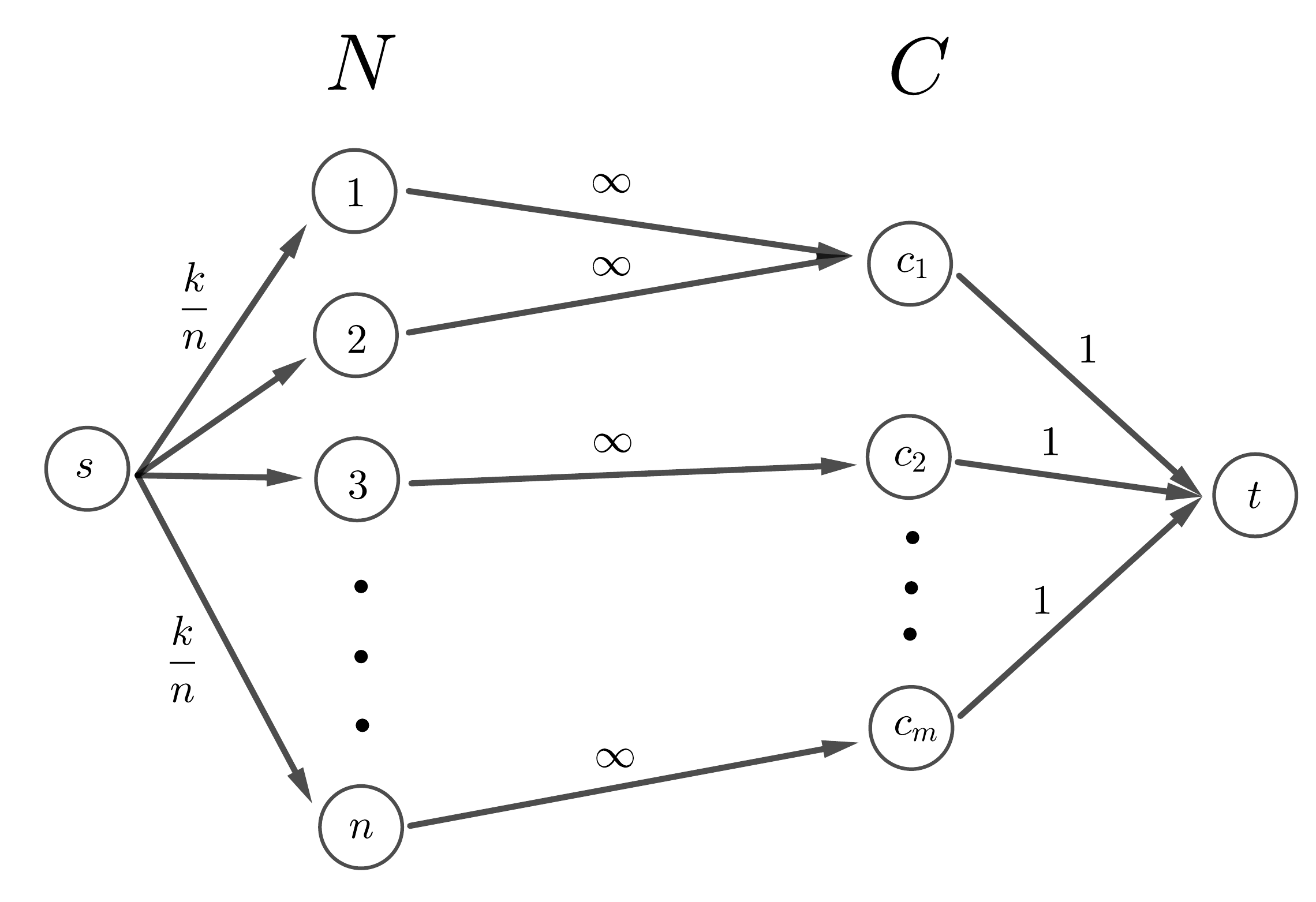}
	\end{center}
	\caption{Illustration of the network representation of an instance of committee voting.}\label{fig:1}
\end{figure}

\begin{definition}
	The \emph{network representation}~ of an instance $\instance$, denoted $\network_\instance$, is a flow network with source $s$, sink $t$, a node for each voter and each candidate, and edge capacities defined as follows (an edge exists if and only if a capacity is defined below):
	\begin{itemize}
		\item $\capacity(s,i) = k/n \quad \forall i\in N$
		\item $\capacity(i,c) = \infty \quad \forall i\in N, c\in A_i$
		\item $\capacity(c,t) = 1 \quad \forall c\in C$.
	\end{itemize}
\end{definition}
With the exception of \Cref{sec:gcut}, we will consider a single instance at a time, and hence simply use $\network$ to denote $\network_\instance$.
For $T\subseteq C$, we will define $\network_\instance(T)$ to be the network $\network_\instance$ where vertices corresponding to candidates in $C\setminus T$ and associated arcs are removed.
Formally, if $\instance = (C, \mathcal{A}, k)$, then $\network_\instance(T) = \network_{I'}$ where $I' = (T,\mathcal{A}', k)$ and $\mathcal{A}' = (A_1\cap T, \ldots, A_n\cap T)$. 
In this paper, we will often search for fair committees using the  network flow representation of our instance.
To do so, we must also be able to translate from network flows to fractional committees.
Given a feasible flow $f$ on the network representation of an instance, we define the \emph{fractional committee given by} $f$ to be $\vec{p}$ where $p_c=f(c, t)$ for all $c\in C$.
We remark that any fractional committee given by a feasible flow must also be feasible.
This holds since $p_c = f(c, t)\leq \capacity(c,t) = 1$ for all $c\in C$, and $\sum_{c\in C} f(c, t) = \sum_{i\in N} f(s, i) \leq k$.
Lastly, we denote the value of a feasible flow $f$ by $val(f)$.

In \Cref{sec:gcut}, we will consider some networks in which the arcs leaving the source need not have a capacity of $k/n.$
Hence, we will formalize the broader class of \emph{entitlement networks}, which have identical structure to network representations of instances but only constrain the arcs exiting the source to have capacities that sum to the instance's committee size $k$.
That is, a network $\wnetwork$ with arc capacities $\capacity$ is an entitlement network if there is an instance $\instance=(C,\mathcal{A},k)$ such that (1) $\sum_{\{v:(s,v)\in \wnetwork\}} \capacity(s,v) = k$ and (2) setting $\capacity(s,v)=k/n$ for all $\{v:(s,v)\in \wnetwork\}$ results in the network representation $\network_\instance$.
We point out that each entitlement network corresponds to exactly one instance of our problem. 
For this reason, we will often refer to the entitlement network nodes connected to the source as voters, $N$, and the nodes connected to the sink as candidates, $C$.
Lastly, abusing notation slightly, we denote the value of maximum flow on some network $\wnetwork$ as $val(\wnetwork).$

\section{A Characterization of Fair Committees}
\label{sec:grp}
In this section, we introduce a new axiom which unifies and strengthens existing notions of fairness. 
We then give a characterization of our fairness notion using the network representation of the problem. 
We first give an intuition and motivation behind our fairness notion by taking inspiration from fractional core.

\begin{definition}[Fractional core]	\label{def:core}
	A fractional committee~$\vec{p}$ satisfies \emph{fractional core} if there is no group of voters $S$ and committee $\vec{q}$ with $\sum_{c\in C} q_c \leq |S|\frac{k}{n}$ such that $u_i(\vec{q})> u_i(\vec{p})$ for all $i\in S$.
\end{definition}

The alternative fractional committee $\vec{q}$ in the above definition is referred to as a \emph{blocking deviation}.
By weakening the definition of blocking deviations to allow the utility of some voters in $S$ to be unchanged by the deviation, we can define a stronger concept which we call \emph{strict fractional core}.
While both versions of fractional core coincide in the single-winner setting, this is not the case here. 
In fact, the strict fractional core may be empty whereas a fractional core outcome always exists. 

 Each of these fractional core concepts has a \textit{fair taxation} interpretation \cite{MSW22,Lind58}, which we borrow from to motivate our own fairness notion.
 The quantity $\frac{k}{n}$ can be thought of as the tax contribution of a voter. 
 The tax money collected is used to buy candidates, where each candidate has a unit cost.\footnote{Candidates can be bought fractionally, meaning that if in total $\alpha \in (0,1)$ dollars are spent on a candidate $c$, then the candidate is selected fractionally with $p_c=\alpha$ .} 
 A committee satisfies fractional core if no group of voters $S$ could collectively use their tax contribution $|S|\frac{k}{n}$ to ``buy'' a fractional committee which they all prefer.
 Fractional core captures the notion that no group of voters, using only their own resources, can achieve a better outcome for themselves. 
 By comparison, our fairness notion states that the total amount of tax contribution (resources) spent on candidates approved by $S$ must be at least the maximum amount that the group $S$ could spend, provided that no voter's contribution is spent on a candidate they do not approve. 
 
\begin{definition} [Group Resource Proportionality]
A fractional committee $\vec{p}$ is said to satisfy \textit{group resource proportionality} (GRP) if for every $S\subseteq N$: 
\begin{equation}\label{eq:GRP}
	\sum\limits_{c\in \bigcup\limits_{i\in S} A_i} p_c \geq |S|\frac{k}{n} - \max\limits_{T\subseteq S}  \left[ |T|\frac{k}{n} - |\bigcup_{i\in T }A_i|\right].
\end{equation}
\end{definition}

It is important to note that the right-hand side of inequality~(\ref{eq:GRP}) is not necessarily $|S| \frac{k}{n}$, which one might expect, considering the fair taxation interpretation in which each voter in $S$ contributes $\frac{k}{n}$.
To see why this is not the case, consider a scenario in which an $\alpha$-fraction of the voters collectively approve of strictly less than $\alpha\cdot k$ candidates. 
In this case, a right-hand side of $|S|\frac{k}{n}$ would enforce that $\alpha\cdot k$ probability be allocated to candidates approved by some member of this group of voters, which is impossible.
Thus, to ensure existence, we subtract from the right-hand side a ``penalty'' measuring the gap between the group's proportional entitlement and the number of candidates they collectively approve, taken maximally over all subsets of voters in $S$.
Through this construction, \GRP gives to every group of voters a representation guarantee that avoids the shortfalls exhibited by fractional core in \Cref{ex:core_bad}.

Our fairness notion has a natural interpretation in terms of lotteries. Any lottery $\Delta$ implementing a GRP fractional committee $\vec{p}$  satisfies, for every group of voters $ S \subseteq N$,
 $$
 \mathbb{E}_{W\sim \Delta}[ | \{c \in W \ | \ \exists i\in S \text{ such that } c\in A_i \} | ] \geq |S|\frac{k}{n} - \max\limits_{T\subseteq S}  \left[ |T|\frac{k}{n} - |\bigcup_{i\in T }A_i|\right],
 $$
 where the left hand side is the expected number of selected candidates approved by some voter in $S$. 
 Thus, \GRP lotteries give a strong ex-ante representation guarantee to every coalition of voters.
  
 \subsection*{GRP Characterization }

  As \GRP gives strong representation guarantees for every coalition of voters, it is not clear a priori whether such fractional committees always exist or  how to obtain one.  In the following result, we use the  network flow formulation of our problem to provide a characterization of \GRP outcomes. 

\begin{theorem} \label{thm:gr_characterization}A fractional committee $\vec{p}$ satisfies group resource proportionality if and only if there exists a max flow $f$ on the network representation  $\mathcal{N}$  such that $p_c \geq  f(c , t)$ for each $c \in C$.
\end{theorem}
\begin{proof}
Let $\vec{p}$ be a fractional committee satisfying \GRP, and consider a modified network $\widetilde{\mathcal{N}}$ with modified capacities $\capacity(c , t) = p_c$ for each $c\in C$. 
By the max flow min-cut theorem, there exists a maximum (or max) flow $f$ on $\widetilde{\mathcal{N}}$ such that 
\begin{align*}
	\sum_{c\in C} f(c, t) &= \min_{T\subseteq N} \left[ |N\setminus T| \frac{k}{n} + \sum\limits_{c\in \bigcup\limits_{i\in T} A_i} p_c \right] \\ 
	&=  |N\setminus T^*|\frac{k}{n} + \sum\limits_{c\in \bigcup\limits_{i\in T^*} A_i} p_c \\ 
	&\geq |N\setminus T^*|\frac{k}{n} + |T^*|\frac{k}{n} - \max\limits_{T\subseteq T^*}  \left[ |T|\frac{k}{n} - |\bigcup_{i\in T }A_i|\right] \\ 
	&= k - \max\limits_{T\subseteq T^*}  \left[ |T|\frac{k}{n} - |\bigcup_{i\in T }A_i|\right] \\ 
	&\geq k - \max\limits_{T\subseteq N}  \left[ |T|\frac{k}{n} - |\bigcup_{i\in T }A_i|\right]
\end{align*}
where the first inequality follows from the fact that $\vec{p}$ satisfies \GRP. 
Observe that $$k - \max\limits_{T\subseteq N}  \left[ |T|\frac{k}{n} - |\bigcup\limits_{i\in T }A_i|\right]= \min\limits_{T\subseteq N} \left[ |N\setminus T| \frac{k}{n} + | \bigcup\limits_{i\in T} A_i| \right] $$ which is equal to the min cut value of $\mathcal{N}$. 
Hence, $f$ is indeed a max flow on $\mathcal{N}$ and since $f$ is also a valid flow on $\widetilde{\mathcal{N}}$, it satisfies  $p_c \geq  f(c , t)$ for each $c \in C$. 

We now prove the converse direction.  Let $f$ be a max flow on network formulation $\mathcal{N}$ and $p_c \geq f(c, t)$ for each $c\in C$.
Suppose, on the contrary, that $\vec{p}$ does not satisfy \GRP. That is, there exists $S'\subseteq N$ such that
\begin{align}\label{ineq:char}
	\sum\limits_{c\in \bigcup\limits_{i\in S'} A_i} p_c< |S'|\frac{k}{n} - \max\limits_{T\subseteq S'}  \left[ |T|\frac{k}{n} - |\bigcup_{i\in T }A_i|\right].
\end{align}
Observe that  $|S'|\frac{k}{n} - \max\limits_{T\subseteq S'}  \left[ |T|\frac{k}{n} - |\bigcup_{i\in T }A_i|\right]= \min\limits_{T\subseteq S'} \left[ |S'\setminus T| \frac{k}{n} +|\bigcup_{i\in T }A_i|   \right] $, which is the min cut value of a subnetwork $\mathcal{N'}$ where the set of voters is $S'$ and the set of candidates is $\bigcup\limits_{i\in S'} A_i$. By the max flow min-cut theorem, we know that there exists a max flow $f'$ on  $\mathcal{N'}$ whose value satisfies $ \sum\limits_{c\in \bigcup\limits_{i\in S'} A_i} f'(c , t) = \min\limits_{T\subseteq S'} \left[ |S'\setminus T| \frac{k}{n} +|\bigcup_{i\in T }A_i|   \right] $. Consider a new flow $f^*$ on the entire network $\mathcal{N}$ defined as follows:
\begin{align*}
	&f^*(s , i) = f'(s , i) \ \ \forall i\in S', \ \ &f^*(i , c_j) = f'(i , c_j) \ \ \forall i\in S', \forall c_j\in A_i \\ 
	&f^*(c_j , t)=f'(c_j , t) \ \ \forall c_j\in \bigcup\limits_{r\in S'} A_r \ \ & \\
	&f^*(s , i) = f(s , i) \ \ \forall i\in N\setminus S' \ \ & f^*(i , c_j)=0 \ \  \forall i\in N\setminus S', \ \forall c_j\in \bigcup\limits_{i\in S'} A_i \\ 
	&f^*(i , c_j)=f(i , c_j) \ \  \forall i\in N\setminus S', \ \forall c_j\in A_i\setminus \bigcup\limits_{r\in S'} A_r \ \ & f^*(c_j , t)=f(c_j , t) \ \ \forall c_j\in C\setminus \bigcup\limits_{r\in S'} A_r 
\end{align*}
Note that $f^*$ satisfies the capacity constraints as well as the flow conservation constraints on the entire network $\mathcal{N}$, and thus is a valid flow on $\mathcal{N}$. Finally we see that, 
\begin{align*}
	\sum_{c\in C} f(c, t) &=  \sum\limits_{c\in \bigcup\limits_{i\in S'} A_i} f(c , t) + \sum\limits_{c\in C\setminus \bigcup\limits_{i\in S'} A_i} f(c , t) \\
	& \leq  \sum\limits_{c\in \bigcup\limits_{i\in S'} A_i}  p_c + \sum\limits_{c\in C\setminus \bigcup\limits_{i\in S'} A_i} f(c , t) \qquad \because f(c_i)\leq p_i \ \  \forall c_i\in C \\
	&< \left[ |S'|\frac{k}{n} - \max\limits_{T\subseteq S'}  \left[ |T|\frac{k}{n} - |\bigcup_{i\in T }A_i|\right] \right] + \sum\limits_{c\in C\setminus \bigcup\limits_{i\in S'} A_i} f(c , t) \quad \because by \ \ (\ref{ineq:char}) \\ 
	&=  \min\limits_{T\subseteq S'} \left[ |S'\setminus T| \frac{k}{n} +|\bigcup_{i\in T }A_i|   \right] + \sum\limits_{c\in C\setminus \bigcup\limits_{i\in S'} A_i} f(c , t)\\
	&= \sum\limits_{c\in \bigcup\limits_{i\in S'} A_i} f'(c , t) + \sum\limits_{c\in C\setminus \bigcup\limits_{i\in S'} A_i} f(c , t) \\ 
	&= \sum_{c\in C} f^*(c, t)
\end{align*}
However, this contradicts the assumption that $f$ is a max flow on $\mathcal{N}$.  
\end{proof}

An immediate consequence of \Cref{thm:gr_characterization} is that every fractional committee given by a max flow $f$ on the network representation $\mathcal{N}$ satisfies \GRP.  
As we will see, the characterization result also provides a useful tool for designing rules which satisfy \GRP.

\subsection*{Comparison With Existing Fairness Notions}

In this section, we compare our fairness notion with the existing fairness notions. 
We begin by noting that, due to the structural differences between the two settings, significant challenges arise when extending axioms and algorithms from the single-winner setting of probabilistic voting to probabilistic committee voting. 
This was explored by \citet{ALS+23}, who extended the fair share axioms from the single-winner setting \citep{Dudd15a,BMS05a} to the committee voting setting, showing that two alternative interpretations are possible resulting in two distinct fairness hierarchies. 
The strongest axioms of the two distinct hierarchies are \emph{group fair share} (GFS) and \emph{strong unanimous fair share} (Strong UFS).
\begin{definition}[Strong UFS \citep{ALS+23}]
	\label{def:sUFS}
	A fractional committee~$\vec{p}$ satisfies \emph{Strong UFS} if for all~$S \subseteq N$ where $\approvalBallotOfAgent{i} = \approvalBallotOfAgent{j}$ for any~$i, j \in S$, it holds for each $i\in S$ that $u_i(\vec{p}) = \sum_{c \in \approvalBallotOfAgent{i}} p_c \geq \min \left\{ |S|  \frac{\committeeSize}{n}, |\approvalBallotOfAgent{i}| \right\}.$
\end{definition}

\begin{definition}[GFS \citep{ALS+23}]
	\label{def:GFS}
	A fractional committee~$\vec{p}$ satisfies \emph{GFS} if it holds for every~$S \subseteq N$ that $\sum_{c \in \bigcup_{i \in S} \approvalBallotOfAgent{i}} p_c \geq \frac{1}{n} \cdot \sum_{i \in S} \min\{\committeeSize, |\approvalBallotOfAgent{i}|\}.$
\end{definition}

We first show that our fairness axiom unifies both fairness hierarchies of \cite{ALS+23} by strengthening both GFS and Strong UFS.

\begin{restatable}{proposition}{propufsgfs} \label{prop:UFS_GFS} 
	Group resource proportionality implies Strong UFS and GFS.
\end{restatable}
\begin{proof}
	We first show that \GRP implies Strong UFS. Consider any  group  of voters $S\subseteq N$ with $\approvalBallotOfAgent{i} = \approvalBallotOfAgent{j}$ for any~$i, j \in S$. For any voter $i\in S$, we have that
	\begin{align*}
		u_i(\vec{p}) = \sum_{c\in A_i} p_c &\geq |S| \frac{k}{n} - \max\limits_{T\subseteq S}  \left[ |T|\frac{k}{n} - |\bigcup_{j\in T }A_j|\right]  \  \qquad \because \vec{p} \text{ satisfies \GRP} \\
		& = |S| \frac{k}{n} - \max\limits_{T=\emptyset, T=S}\left[ |T|\frac{k}{n} - |\bigcup_{j\in T } A_j|\right] \quad \because \bigcup_{j\in T } A_j = A_i   \text{ for any } T\neq \emptyset \\ 
		& = |S| \frac{k}{n} - \max\left[|S|\frac{k}{n}-|A_i|,\ 0 \right] \\
		& = \min\left[|S| \frac{k}{n},\ |A_i|\right].
	\end{align*}
	
	Next we show that \GRP implies GFS. Fix arbitrary set of voters $S\subseteq N$. 
	Let $Q=\{ i\in S \ | \ |A_i|\leq k\}$ and $T^*=\argmax\limits_{T\subseteq S} \left[ |T|\frac{k}{n} - |\bigcup_{i\in T }A_i|\right]$, observe that $ T^*  \subseteq Q$. To see this suppose there exists $j\in S\setminus Q$ (and thus $|A_j|\geq k+1$) such that $j\in  T^* $. We have that $|\bigcup_{i\in T^* }A_i| \geq k+1$  but this results in  contradiction as  $0\leq |T^*|\frac{k}{n} - |\bigcup_{i\in T^* }A_i| \leq  |T^*|\frac{k}{n} - (k+1) < 0 $. Finally, we see that, 
	\begin{align*}
		\sum\limits_{c\in \bigcup_{i\in S} A_i} p_i &\geq |S|\frac{k}{n} - \max\limits_{T\subseteq S}  \left[ |T|\frac{k}{n} - |\bigcup_{i\in T }A_i|\right] \qquad\quad \because \vec{p} \text{ satisfies \GRP}  \\ 
		&= |S|\frac{k}{n} - \max\limits_{T\subseteq Q}  \left[ |T|\frac{k}{n} - |\bigcup_{i\in T }A_i|\right] \\ 
		&\geq |S|\frac{k}{n} - \max\limits_{T\subseteq Q}  \left[ |T|\frac{k}{n} - \frac{1}{n}\sum_{i\in T}|A_i|\right] \qquad \because |\bigcup_{i\in T}A_i|\geq  \frac{1}{n}\sum_{i\in T}|A_i|   \\
		& = |S|\frac{k}{n} -\left[ |Q|\frac{k}{n} - \frac{1}{n}\sum_{i\in Q}|A_i|\right] \\ 
		& = |S\setminus Q |\frac{k}{n}+ \frac{1}{n}\sum_{i\in Q}|A_i| \ = \ \frac{1}{n} \sum_{i\in S} \min (k, |A_i|). 
	\end{align*}
\end{proof}

Furthermore, we point out that if an integral committee satisfies GRP, then it also satisfies the well-studied proportional representation notion of \emph{proportional justified representation} PJR \citep[see, e.g., ][]{SEL+17a,AEH+18a,LaSk23a}. 
An analogous result does not hold for Strong UFS or GFS \citep{ALS+23}.
Since PJR is not a focus of this work, we defer its definition and the proof of the following proposition to the appendix.

\begin{restatable}{proposition}{propPJR} \label{prop:grp_implies_pjr}
	If an integral committee satisfies GRP, then it satisfies PJR.
\end{restatable}

We next show that strict fractional core implies GRP. 
This result shows that, when targeting proportional representation, strict fractional core --- despite its non-existence --- is a more appropriate generalization of core (than fractional core) from the $k=1$ special case.

\begin{restatable}{proposition}{propcore} \label{prop:core_implies_gr}
	If a fractional committee satisfies strict fractional core, then it satisfies group resource proportionality.
\end{restatable}
\begin{proof}
	We prove the contrapositive.
	Let $\vec{p}$ be a fractional committee which does not satisfy \GRP and let $S\subseteq N$ be a group of voters for which \GRP is violated.
	We will show the existence of a blocking deviation $\vec{q}$ for the group of voters $S$ and thus show that $\vec{p}$ does not satisfy strict fractional core.
	
	If $|\bigcup_{i\in S} A_i| \leq |S|\frac{k}{n}$, then let $\vec{q}=\vec{1}_{\cup_{i\in S} A_i}$.
	It is clear that $\sum_{c\in C} q_c\leq |S|\frac{k}{n}$ and that $u_i(\vec{p})\leq u_i(\vec{q})$ for all $i\in S$.
	Since S violates GRP, we have
	\begin{align*}
		\sum\limits_{c\in \bigcup\limits_{i\in S} A_i} p_c < |S|\frac{k}{n} - \max\limits_{T\subseteq S}  \left[ |T|\frac{k}{n} - |\bigcup_{i\in T }A_i|\right] \leq |\bigcup_{i\in S} A_i|
	\end{align*}
	and thus there must be some $i\in S$ for which $u_i(\vec{p})<u_i(\vec{q})$. 
	
	Suppose instead that $|\bigcup_{i\in S} A_i| > |S|\frac{k}{n}$. 
	Note that
	\begin{align*}
		\sum\limits_{c\in \bigcup\limits_{i\in S} A_i} p_c < |S|\frac{k}{n} - \max\limits_{T\subseteq S}  \left[ |T|\frac{k}{n} - |\bigcup_{i\in T }A_i|\right] \leq |S|\frac{k}{n} < |\bigcup_{i\in S} A_i|
	\end{align*}
	and thus there must be some candidate $c'\in\bigcup_{i\in S} A_i$ such that $p_{c'}<1$. 
	Let $\epsilon=\min(|S|\frac{k}{n} - \sum_{c\in \bigcup\limits_{i\in S} A_i} p_c, 1-p_{c'})$.
	Set $q_{c'} = p_{c'} + \epsilon$, $q_c=p_c$ for all $c\in\bigcup_{i\in S} A_i \setminus \{c'\}$, and $q_c=0$ otherwise.
	By construction, $\sum_{c\in C} q_c \leq |S|\frac{k}{n}.$ Furthermore, it is clear that $u_i(\vec{p})\leq u_i(\vec{q})$ for all $i\in S$ and the voters in $S$ who approve of $c'$ will strictly improve.  
\end{proof}

\begin{figure}[t]
\centering
\scalebox{1.0}{
\begin{tikzpicture}
\tikzstyle{onlytext}=[]
\tikzset{venn circle/.style={circle,minimum width=0mm,fill=#1,opacity=0.1}}

\node[onlytext] (strict fractional core) at (0,0) {\begin{tabular}{c}{strict fractional core*}\end{tabular}};

\node[onlytext,text=green!70!black!90] (GRP) at (-3,-1.5) {\begin{tabular}{c}{\textbf{GRP}}\end{tabular}};
\node[onlytext] (fractional core) at (3,-1.5) {\begin{tabular}{c}{fractional core}\end{tabular}};

\node[onlytext,text=green!70!black!90] (Strong UFS) at (-3,-3) {\begin{tabular}{c}{Strong UFS}\end{tabular}};
\node[onlytext,text=green!70!black!90] (GFS) at (3,-3) {\begin{tabular}{c}{GFS}\end{tabular}};

\draw[->, line width=1pt] (strict fractional core) -- (GRP);
\draw[->, line width=1pt] (strict fractional core) -- (fractional core);

\draw[->, line width=1pt] (fractional core) -- (GFS);
\draw[->, line width=1pt] (fractional core) -- (Strong UFS);

\draw[->, line width=1pt] (GRP) -- (GFS);
\draw[->, line width=1pt] (GRP) -- (Strong UFS);
\end{tikzpicture}
}
\caption{Visualization of hierarchy of fairness properties mentioned in this paper.
An arrow from (A) to (B) denotes that (A) implies (B).
Properties highlighted in green are polynomial time computable, following from \Cref{thm:gr_characterization}.
(*) Strict fractional core is not guaranteed to exist.
}
\label{fig:relations}
\end{figure}

Refer to \Cref{fig:relations} for a visualization of the logical relations of each of the described fairness properties.
Although \GRP is implied by the strict fractional core when it exists, GRP and fractional core are logically incomparable (see \Cref{ex:core_bad}). 
There are several advantages to our fairness notion as opposed to fractional core.
Perhaps the most apparent is that, while no known efficient algorithm exists for the fractional core \citep{MSW22}, our fairness notion can be achieved through an efficient algorithm. 
Second, given a fractional committee $\vec{p}$, checking whether it satisfies \GRP can be done in polynomial time (\Cref{prop:check}). 
This is in contrast to the fractional core, for which, to the best of our knowledge, it is not known whether such an algorithm exists. 

\begin{restatable}{proposition}{propcheck} \label{prop:check} 
	For a given fractional committee $\vec{p}$, checking whether $\vec{p}$ satisfies group resource proportionality can be done in polynomial time.
\end{restatable}
\begin{proof}
	Given fractional committee $\vec{p}$, one can check whether $\vec{p}$ satisfies \GRP by the following simple procedure.
	Let $\network$ be the network formulation of the instance.
	Modify $\network$ so the capacities on arcs from the candidate set to the sink are $\capacity(c, t) = p_c$ for all $c\in C$, and denote the resulting network $\network'$.
	Compute a max flow $f'$ on $\network'$ and max flow $f$ on $\network$. 
	Note that, because of the capacity constraints in $\network'$, the fractional committee given by $f'$ is elementwise dominated by $\vec{p}$.
	Finally, since the flow values of $f$ and $f'$ will be equal if and only if $f'$ is a max flow on $\network$, it follows from \Cref{thm:gr_characterization} that $\vec{p}$ is \GRP if and only if the flow values of $f$ and $f'$ are equal.
\end{proof}

\section{Fairness and Efficiency: Redistributive Utilitarian Rule}
\label{sec:efficiency}
Given that there may be many max flow solutions to the network formulation of a given instance, and thus many solutions satisfying our fairness notion, we seek to further refine the set of desirable outcomes. 
We begin by searching for outcomes which are both fair and (Pareto) efficient, meaning the resulting utility profile is not strictly Pareto dominated by another outcome's utility profile.

In the single-winner setting, the voting rule which maximizes Nash welfare is ex-ante efficient and satisfies fractional core.
However, as we will now show, the maximum Nash welfare rule does not guarantee either fractional core \emph{or} \GRP in the probabilistic committee voting setting.
\begin{proposition}\label{ex:nash} Maximum Nash welfare rule violates \GRP and fractional core
	\end{proposition}
	\begin{proof}
	
Consider an instance with four voters, four candidates, $k=2$, and the voters' preferences are: $A_1=\{a,d \}, A_2 = A_3= \{a,b\}$, and $A_4 = \{c\}$.
	Since $N_a\supsetneq N_b$, $a$ dominates $b$ with respect to Nash welfare and thus $b$ can be selected with positive probability in the Nash welfare-maximizing outcome if and only if $a$ is integrally selected.
	This also holds for $a$ with respect to $d$ and thus, due to the committee size and number of candidates, it is clear that $a$ is integrally selected.
	From this point, one can observe that Nash will give additional probability to $c$ instead of $d$ since each candidate only increases the utility of one voter and voter $4$'s utility will necessarily be less than that of voter $1$. 
	One can then frame Nash welfare as a single-variable optimization to show that $\vec{p} = ( 1,\frac13, \frac23, 0 )$ is the unique fractional committee maximizing Nash welfare.
	
	However, denoting $S=\{1,2,3\}$, we can then see that $\vec{p}$ violates \GRP with respect to $S$: 
	$$|S|\frac{k}{n} - \max\limits_{T\subseteq S}  \left[ |T|\frac{k}{n} - |\bigcup_{i\in T }A_i|\right] = \frac32 > \frac43 =  \sum_{j\in \cup_{i\in S} A_i} p_j  $$
	where the first equality follows since the maximum is attained with the empty set. 
	Similarly, we can show $\vec{p}$ violates fractional core by showing that a committee using total probability $3/2$ can increase the utility of every voter in $S$.
	In particular, any committee $\vec{q} = (1, \frac{1}3 + \epsilon, 0, \delta)$ with $\delta + \epsilon \leq \frac{1}6$ satisfies the size constraint and strictly benefits every voter in $S$.

\end{proof}
It is worth pointing out that the incompatibility proved by \Cref{ex:nash} does not appear to result from a clash between efficiency and fairness.
Indeed, the natural \GRP fractional committee $\vec{q} = ( 1, \frac12, \frac12,0)$  is ex-ante efficient.
Furthermore, the fractional committee selected by the Nash rule in \Cref{ex:nash} is more egalitarian than any solution satisfying \GRP.\footnote{Note that any fractional committee $\vec{q}$ with $q_c>\frac12$ cannot elementwise dominate a max flow on $\network$.}
Thus, Nash welfare embodies a sense of fairness distinct from that of fractional core and similar concepts, which measure fairness by comparing against a ``deserved'' outside option, as opposed to taking a welfarist approach. 
While these fairness notions coincide in the single-winner setting, they diverge in the committee voting setting.

Since we already know how to compute a fractional committee satisfying \GRP in polynomial time, it is natural to ask whether we can achieve fairness and efficiency by iteratively identifying and applying Pareto improvements to a fair committee, which can be done with a simple linear program.
However, even in the single-winner special case, an arbitrary Pareto improvement with respect to a \GRP fractional committee need not maintain the weaker property of GFS, let alone \GRP.\footnote{Suppose there are four candidates and three voters with preferences as follws: $A_1=\{a,b\}, A_2 = \{b,c\}, A_3 = \{d\}$. It can be verified that $\vec{q} = ( \frac13, 0, \frac 13, \frac 13 )$ satisfies \GRP. 
Now consider the fractional committee $\vec{p} = ( 0, \frac13, 0, \frac23 )$, which Pareto dominates $\vec{q}$ but does not satisfy GFS with respect to voter group $\{1,2\}.$}
Thus, an approach which treats fairness and efficiency sequentially is likely to meet with significant obstacles.
In the next subsection, we will present an algorithm that instead maintains the invariant of efficiency while constructing a fair fractional committee.

\begin{algorithm}[t]
\caption{Redistributive Utilitarian Rule (\RUT)}
\label{alg:rut}
\DontPrintSemicolon

\KwIn{Voters~$N$, candidates~$C$, approval profile~$\mathcal{A}=(\approvalBallotOfAgent{i})_{i \in N}$ and committee size $k$.}
\KwOut{A fractional committee~$\vec{p} = (\Delta_c)_{c \in C}$ of size $k$.}

$\lambda_i\gets 1$ for all $i\in N$\;
$s^*\gets \max_{c\in C} s_\lambda(c)$\;
Let $f_0$ denote a trivial flow on any flow network.\;
$j\gets 1$\; 
\While{$j\leq m$}{
	$c_j\gets \argmax_{c\in C} s_\lambda(c)$\;
	$\network^j\gets \network(\{c_1,\ldots,c_j\})$ \;
	Apply \Cref{lem:circulation} to the flow $f_{j-1}$ on $\network^j$ to obtain a max flow $f_j$ on $\network^j$.\;
	$V_j\gets \{i\in N: f_j(s, i) < \frac{k}n\}$\;
	\If{$A_i\subseteq \{c\in C: f_j(c, t)=1\} \ \forall i\in V_j$}
	{ Exit loop.}
	$\alpha\gets \min \left \{\alpha\in \mathbb{R}: \max\limits_{c\in C\setminus \{c_1,\ldots,c_j\}} \alpha \cdot \vert N_c\cap V_j \vert + s_\lambda(c) = s^* \right\}$\;
	\For{$i\in V_j$}{
		$\lambda_i\gets \lambda_i + \alpha$\;
	}
	$j\gets j+1$
}
$f^*\gets f_j$\;
Denote $\vec{p}$ as the fractional committee given by $f^*$.\;
\While{$k-\sum_{c\in C} p_c>0$}{
	$\delta\gets k-\sum_{c\in C} p_c$\;
	$c\gets \argmax\limits_{\{c\in C: p_c<1\}} s_\lambda(c)$\;
	$p_c\gets \max (1, p_c + \delta)$\;
}
\BlankLine
\Return{$\vec{p}$}
\end{algorithm}
At a high level, our algorithm -- which we call the \emph{redistributive utilitarian rule} (\RUT) -- maintains the invariant of efficiency while constructing a fair fractional committee.
To do so, it ensures that the selected committee maximizes weighted utilitarian welfare while iteratively computing max flows and redistributing flow to avoid long augmenting paths.
We will now give a more detailed description of \RUT.
Refer to \Cref{alg:rut} for a detailed description in pseudocode.

We start with unit weights $\lambda_i=1$ and a modified version of $\network$ such that all candidates are removed from the network.
We identify some candidate $c^*$ which maximizes over all $c\in C$ the ``score'' $s_\lambda(c)=\sum_{i\in N_c} \lambda_i$ and let $s^*=\sum_{i\in N_{c^*}} \lambda_i$.
In each round, we add $c^*$ to our modified network and compute a max flow $f$ which balances flows in a way that avoids saturating any edge from the source to a voter whenever possible (see \Cref{lem:circulation}). 
The weights $\lambda_i$ are then effectively frozen for any voters with $f(s, i)=\frac{k}n$, and the weights of the other voters are uniformly increased until there is some new candidate $c^*$ with $s_\lambda(c^*)=s^*$.

This loop terminates when, for each $i\in N$ with $f(s, i)< \frac{k}n$, it holds that $A_i\subseteq \{c\in C: f(c, t)=1\}$.
Intuitively, this means there is no voter with unsaturated capacity from the source who approves of some candidate with unsaturated capacity to the sink.
We refer to the network flow resulting from this process as $f^*$.
Lastly, letting $\vec{q}$ denote the fractional committee given by $f^*$, we greedily allocate the remaining probability $k-\sum_{c\in C} q_c$ to candidates in descending order of $s_\lambda(c)$, and return the resulting fractional committee, denoted $\vec{p}$.
 
A subtle but crucial detail of \RUT is the way it computes a max flow on the subnetwork in each iteration. 
If the max flow computed saturates an arc to some voter unnecessarily, then the result of the algorithm may admit an augmenting path in the main network and thus will fail \GRP.
Thus, we carefully redistribute flows in each iteration to compute a max flow which only saturates the arc to a voter when strictly necessary.
This can be thought of as an intermediate fairness check executed during each iteration of \RUT.
This redistributive process gives our voting rule its name and will be detailed in the proof of the following technical lemma, which is key to our proof that \RUT satisfies \GRP.
\begin{lemma} \label{lem:circulation}
	Given a network formulation $\network$ and a feasible flow $f$ on $\network(T)$ for some $T\subseteq C$, there exists a polytime computable max flow $f'$ on $\network(T)$ such that each of the following conditions holds:
	\begin{enumerate}
		\item \label{cond:c1} $\forall i\in N$ with $A_i\subseteq T$, $f(s, i)=\frac{k}n \implies f'(s, i)=\frac{k}n$
		\item \label{cond:c2} If the residual network resulting from flow $f'$ on the main network $\network$ admits an augmenting path, then the shortest such augmenting path is of length three.
		\item \label{cond:c3} $f'(c, t) \geq f(c, t)$ \ $\forall c\in T$
	\end{enumerate}
\end{lemma}
\begin{proof}
	Starting with $f$, apply Edmonds-Karp to compute a max flow $f'$ on $\network(T)$.
	Since Edmonds-Karp weakly increases flow on arcs exiting the source, Condition (\ref{cond:c1}) holds so far.
	We now give a procedure which will give us the desired condition without changing the flow value and thus maintaining the max flow property.
	The procedure can be thought of as iteratively rebalancing payments in order to avoid fully saturating connections from the source to voters who approve of some candidate not contained in the subnetwork.
	
	Let $N'=\{i\in N: f'(s, i)=\frac{k}n, A_i\setminus T\neq \emptyset\}$, the set of voters whose connection to the source is saturated under $f'$, but still approve of some candidate not in $T$.
	Search in the residual network resulting from $f'$ on the subnetwork $\network(T)$ for a cycle $Q_i=(s,i_1,c_1,\ldots,i_r,c_r,i,s)$ with \emph{bottleneck} $b>0$ (i.e., the minimum residual capacity on the cycle is strictly positive) where $i\in N'$.
	Update $f'$ by pushing $\frac12 b$ additional flow along $Q_i$.
	Refer to Figure~\ref{fig:cyc} for an illustration of such a flow update along a cycle. 
	When such a cycle no longer exists, return $f'$.

\begin{figure}
	\begin{center}
		\includegraphics[scale=0.2]{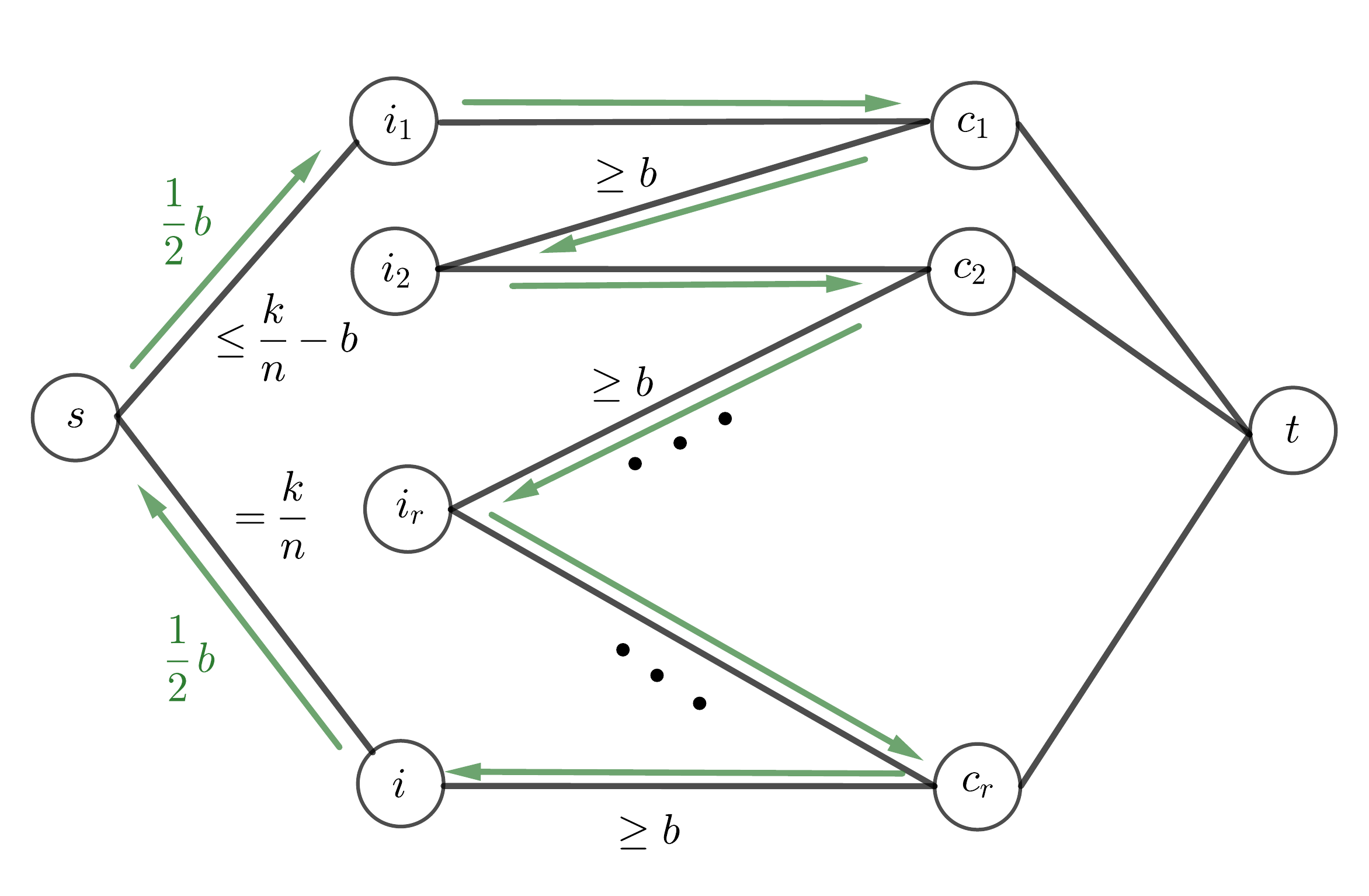}
	\end{center}
	\caption{Illustration of a flow update along a cycle $Q_i=(s,i_1,c_1,\ldots,i_r,c_r,i,s)$ with bottleneck $b>0$. Labels in black denote flow bounds sufficient to trigger the flow update. Green arcs show the path of flow redistribution.}\label{fig:cyc}
\end{figure}
	Note that the flow value of $f'$ is invariant in this cycle  described in the previous paragraph since flow travels in a cycle starting and ending at the source.
	Thus, $f'$ remains a max flow on $\network(T)$. 
	To see that this procedure is polynomial time computable, observe that identification of a cycle in the residual network with positive bottleneck and ending in a voter in $N'$ can be computed with BFS, and that there are at most $n$ iterations.
	The latter observation is true since each cycle decreases the size of $N'$ by exactly one. 
	Specifically, after pushing flow along cycle $Q_i$, $f'(s, i) = \frac{k}n - \frac12 b < \frac{k}n$ (so $i$ no longer belongs to $N'$) and $f'(s, i_1) \leq (\frac{k}n - b) + \frac12 b < \frac{k}n$ (so $i_1$, the only other voter whose flow from the source changes, is still not in $N'$). 
	We can see that Condition (\ref{cond:c1}) is retained throughout this procedure by noting that $A_i\not \subseteq T$ since $i\in N'$ and $f(s, i_1) \leq f'(s, i_1) < \frac{k}n$, and thus the only two voters for whom $f'(s, i)$ has changed still uphold the condition. 

	Before proving Condition (\ref{cond:c2}), we first point out that any feasible flow on $\network(T)$ is also feasible on $\network$ by virtue of $\network(T)$ being a subnetwork.
	This is important as we assume that $f'$ can be treated as a network flow on $\network$ in the statement.
	We prove Condition (\ref{cond:c2}) by contradiction. 
	Suppose that the shortest augmenting path on the residual network $\mathcal{R}_{f'}$ resulting from network flow $f'$ on $\network$ has a length of four or more.
	Then, the shortest augmenting path is of the form $P=(s,i_1,c_1,\ldots,i_r,c_r,t)$ for some $r\geq 2.$ 
	Note that $f'(s, i_r) = \frac{k}n$ since otherwise this edge would have residual capacity and $(s,i_r,c_r,t)$ would constitute an augmenting path of length three.
	Also note that $c_r\not\in T$ since otherwise $P$ would constitute an augmenting path in the residual network from network flow $f'$ on $\network(T)$, contradicting that $f'$ is a max flow on $\network(T)$.
	Together, these two observations show that $i_r\in N'$ at the termination of the loop.
	
	It is apparent that $\{c_1,\ldots,c_{r-1}\}\subseteq T$ since there is residual capacity on the ``backward'' arc $(c_j,i_{j+1})$ in $\mathcal{R}_{f'}$ for all $j\in[r-1]$.
	This means that $Q=(s,i_1,c_1,\ldots,i_r, s)$ is a valid cycle in the residual network resulting from flow $f'$ on $\network(T)$. 
	Furthermore, since $(i_r,s)$ has strictly positive residual capacity as does every other arc in $Q$ (by virtue of belonging to an augmenting path $P$), it holds that the bottleneck of $Q$ is strictly positive.
	This leads us to a contradiction since the termination condition of our procedure was not met.
	Condition (\ref{cond:c3}) holds since Edmonds-Karp weakly increases flows on each arc from a candidate to the sink and the procedure we execute afterwards never alters the flows from a candidate to the sink.
\end{proof}
We can now state and prove the main theorem of the section.
\begin{theorem}
	Redistributive Utilitarian Rule computes a fractional committee satisfying efficiency and group resource proportionality in polynomial time.
\end{theorem}
\begin{proof}
	We first point out that the first while loop (which we will refer to as the main loop) of \RUT (\Cref{alg:rut}) will always terminate due to the if-condition. 
	This is because, when $j=m$, $\network^j = \network$ and thus applying \Cref{lem:circulation} computes a max flow on $\network$. 
	And if any voter in $V^j$ approved of some candidate that was not integrally funded, there would be an augmenting path.
	Henceforth, we denote the total number of rounds executed in the main loop as $z$.

	\paragraph{\GRP}
	We now show that $\vec{p}$, the fractional committee computed by \RUT, satisfies \GRP.
	Since $\vec{p}$ clearly elementwise dominates the fractional committee given by $f^*$, it is sufficient by \Cref{thm:gr_characterization} to show that $f^*$ constitutes a max flow on $\network$.
	Assume, for the sake of a contradiction, that $f^*$ is not a max flow on $\network$.
	Then, there must be an augmenting path in the residual network $\mathcal{R}_{f^*}$ from network flow $f^*$ on $\network$.
	However, since $f^*$ was computed by an application of \Cref{lem:circulation}, we know that the shortest augmenting path in $\mathcal{R}_{f^*}$ is of length three.
	That is, there is an augmenting path of the form $P=(s,i,c,t)$ for some $i\in N, c\in C$.
	Since there must be residual capacity in $\mathcal{R}_{f^*}$ for each arc in $P$, it is clear that $f^*(s, i)< \capacity(s, i) = \frac{k}n$ and $f^*(c, t)<\capacity(c, t)=1$.
	However, since $(i,c)$ being an arc in $\network$ implies $c\in A_i$, this means that voter $i$ and candidate $c$ contradict the termination condition of the main loop of \RUT.

	\paragraph{Efficiency} 
	We will prove efficiency by showing that the fractional committee $\vec{p}$ is equivalent to a greedy selection of candidates in order of their score $s_\lambda(c)$.
	Let $C^+$ denote the candidates which receive flow during the main loop of \RUT, i.e. $f^*(c, t)>0$ for all $c\in C^+$.
	Let $C^-$ denote the complement set of candidates, i.e. $f^*(c, t)=0$ for all $c\in C^-$. 
	It is clear that all candidates in $C^+$ are included in $\network^z$, the state of the subnetwork when the main loop terminates.
	By construction, since weights weakly increase over the course of \RUT and a candidate $c_j$ is added to the subnetwork $\network^j$ only when it reaches the score $s_\lambda(c_j)=s^*$, we have that $s_\lambda(c)\geq s^*$ for all $c\in C^+$.

	Now suppose there is some candidate $c'$ with $s_\lambda(c')>s^*.$
	Clearly, since weights only increase during the main loop, $s_\lambda(c')$ hit $s^*$ before the main loop terminated and $c'$ was added to the subnetwork.
	Thus, $c'=c_j$ for some $j$.
	Since the score of $c_j$ continued to increase after this iteration, there must have been some voter $i'$ with $c_j\in A_{i'}$ and $f_{j'}(s, i') < \frac{k}n$ for some $j'\geq j$.
	Since $f_{j'}$ is computed via iterated application of \Cref{lem:circulation}, it is clear from Condition (\ref{cond:c1}) that $f_j(s, i') < \frac{k}n$. 
	Since we know $f_j$ is a max flow on $\network^j$, it must be that $f_j(c', t) = 1 $, since otherwise $(s,i',c',t)$ would constitute an augmenting path.
	Then, by Condition (\ref{cond:c3}) of \Cref{lem:circulation}, we have that $f^*(c', t) = 1.$
	Note that this observation implies that for any $c$ with $f^*(c, t)< 1$, it must be that $s_\lambda(c)\leq s^*$.

	We can now show that any two fractionally selected candidates must have equal scores. That is, 
	\begin{equation} \label{eq:eff1}
		s_\lambda(c_1) = s_\lambda(c_2) \qquad \forall c_1,c_2\in C : p_{c_1}, p_{c_2}\in (0,1)
	\end{equation}
	On the contrary, assume without loss of generality, that $s_\lambda(c_1)>s_\lambda(c_2).$
	If $c_2\in C^+$, then 
	$$
	s_\lambda(c_1) > s_\lambda(c_2) \geq s^* \implies f^*(c_1, t) = 1 > p_{c_1}.$$
	Since $p_c\geq f^*(c, t)$ for all $c$ by construction, this is a contradiction.
	So it must be that $c_2\in C^-$ and thus $f^*(c_2, t)=0$, which means that $c_2$ was fractionally selected during the greedy completion phase (the second while-loop in \Cref{alg:rut}).
	However, this loop selects in order of score, and since $c_1$ is not integrally selected and has a higher score, we reach a contradiction.

	Our second key observation is that the score of any integrally selected candidate is always at least that of any fractionally selected candidate, i.e.
	\begin{equation} \label{eq:eff2}
		s_\lambda(c_1)\geq s_\lambda(c_2) \qquad \forall c_1,c_2\in C : p_{c_2} < p_{c_1} = 1.
	\end{equation}
	To see this, first note that $s_\lambda(c_2)\leq s^*$ by our previous observation since $f^*(c_2, t) \leq p_{c_2} < 1$. 
	If $c_1\in C^+$, then $s_\lambda(c_1) \geq s^*$ and Equation~\ref{eq:eff2} follows.
	Otherwise, $c_1\in C^-$ and was selected entirely in the phase which proceeds greedily by score amongst candidates which have not been integrally selected.
	Since $c_2$ is not integrally selected, it follows that $c_1$ has a higher score, i.e. $s_\lambda(c_1)\geq s_\lambda(c_2)$.

	Our third and final observation confirms that any candidate which receives positive probability under \RUT has a higher score than any candidate which does not, i.e.
	\begin{equation} \label{eq:eff3}
		s_\lambda(c_1)\geq s_\lambda(c_2) \qquad \forall c_1,c_2\in C : p_{c_1} > p_{c_2} = 0.
	\end{equation}
	Note that $c_2\in C^-$. Thus, Equation~\ref{eq:eff3} follows since $s_\lambda(c_1)\geq s^* \geq s_\lambda(c_2)$ if $c_1\in C^+$, and by construction of the greedy completion phase if not.

	Altogether, from Equations~\ref{eq:eff1}-\ref{eq:eff3}, we have that $\vec{p}$ is the fractional committee which maximizes the expression 
	$$\max_{\vec{q}} \sum_{c\in C} q_c s_\lambda(c) =  \max_{\vec{q}} \sum_{c\in C} \sum_{i\in N_c}q_c  \lambda_i = \max_{\vec{q}} \sum_{i\in N} \sum_{c\in A_i} q_c \lambda_i = \max_{\vec{q}} \sum_{i\in N} \lambda_i u_i(\vec{q}) .$$
	Thus, $\vec{p}$ maximizes weighted utilitarian welfare for positive weights $\lambda_i$, and efficiency  follows.

	\paragraph{Polynomial time computation}
	The main loop of \RUT adds a candidate to the subnetwork in each iteration and thus terminates in at most $m$ iterations. 
	One can identify the next candidate to add to the subnetwork in polynomial time by simply checking the necessary uniform weight increase for each candidate, and choosing the candidate which requires the minimum weight increase.
	Computation of the relevant max flow is in polynomial time by \Cref{lem:circulation}.
	The termination condition of the loop is also clearly polynomial time checkable.
	Lastly, the greedy completion step requires only a linear pass over a sorted list of the candidates.
\end{proof}

We note that, when applied to the $k=1$ special case, \RUT is equivalent to the \emph{fair utilitarian rule} defined by \citet{BMS02a}.
However, in the single-winner context, each voter allocates their entire share in a single round, and thus the algorithm is significantly simpler and it is not necessary to use a network flow-based approach. 
\citet{BBP+21a} generalized the fair utilitarian rule to the single-winner setting with arbitrary voter endowments.
We note that \RUT could similarly be extended to a generalization of our setting with endowments.

\section{On Strategyproofness, Fairness, and Efficiency}\label{sec:gcut}

While \RUT is both efficient and fair, it is not strategyproof. 
This is inevitable in light of Theorem 2 of \citet{BBP+21a}, which states that no strategyproof and efficient rule can satisfy \emph{positive share} -- a minimal fairness requirement which guarantees every voter non-zero utility.
In the single-winner setting, when strategyproofness and fairness are viewed as strict requirements, the most attractive known voting rule is the \emph{conditional utilitarian rule} (CUT). 
For each voter $i$, the rule distributes $1/n$ to the candidates in $A_i$ which are approved by the greatest number of voters.
CUT is strategyproof and maximizes utilitarian welfare subject to GFS.
However, as we will now show, any extension of CUT to probabilistic committee voting must lose one of these properties.

We say a rule satisfies \emph{\GRP-efficiency} (or is \emph{\GRP-efficient}) if it always returns a fractional committee which is efficient among \GRP fractional committees. 
We define \emph{GFS-efficiency} analagously.

\begin{proposition} \label{prop:impossibility}
	No voting rule satisfies GRP-efficiency and strategyproofness. Likewise, no voting rule satisfies GFS-efficiency and strategyproofness.
\end{proposition}
\begin{proof}
	Consider an instance with $m=n+1$, committee size $k=2$, and an approval profile $\mathcal{A}$ defined as follows: $A_i= \{c_0, c_i\}$ for all $i\in N.$
	Suppose there is a rule $F$ satisfying \GRP-efficiency.
	Let $\vec{p} = F(\mathcal{A},k)$ be the fractional committee returned by the rule. 
	Note that, since $k=2$, it must be that $p_{c_j}>0$ for some $j\in[n].$
	If $p_{c_0}<1$, then we can shift $\epsilon = \min(p_{c_j}, 1-p_{c_0})$ probability to $p_{c_0}$ from $p_{c_j}$ without losing \GRP.
	This holds since $c_0$ is approved by every voter and thus the LHS of every \GRP constraint is weakly increased by this shift.
	Call the resulting fractional committee $\vec{q}.$
	It is apparent that $u_j(\vec{q}) = u_j(\vec{p})$ and $u_i(\vec{q}) > u_i(\vec{p})$ for all $i\neq j$. Thus, $\vec{q}$ Pareto dominates $\vec{p}.$
	This contradicts \GRP-efficiency, and thus we can assume instead that $p_{c_0} = 1.$

	Since $\sum_{r=1}^n p_{c_r} = k - p_{c_0} = 1$, it follows that $p_{c_j}\leq \frac1n$ for some $j\in[n].$
	Note that $u_j(\vec{p}) = p_{c_0} + p_{c_j} \leq 1 + \frac1n.$
	Now consider an alternative approval profile where voter $j$ misreports her preferences by dropping $c_0,$ i.e. the approval profile $\mathcal{A}' = \{A_1,A_2,\ldots,A'_j,\ldots,A_n\}$ where $A'_j = \{c_j\}.$
	Let $\vec{q}= F(\mathcal{A}',k)$ be the fractional committee returned by the rule in this case.
	By a similar argument to before, it must be that $q_{c_0} = 1.$
	Furthermore, \GRP requires that $\sum_{c\in A'_j} q_c \geq \frac{k}n = \frac{2}n$. 
	Thus,
	\begin{equation*}
		u_j(F(\mathcal{A}',k)) \ = \ u_j(\vec{q}) \ = \ q_{c_0} + q_{c_j} \ \geq \ 1 + \frac{2}n \ > \ 1 + \frac1n \geq u_j(\vec{p}) \ = \ u_j(F(\mathcal{A},k)).
	\end{equation*}
	This proves that $F$ is not strategyproof. The same argument can be used to show that GFS-efficiency and strategyproofness are likewise incompatible.
\end{proof}

\begin{algorithm}[t]
\caption{Generalized CUT}
\label{alg:gcut}
\DontPrintSemicolon

\KwIn{Voters~$N$, candidates~$C$, approval profile~$\mathcal{A}=(\approvalBallotOfAgent{i})_{i \in N}$, and committee size $k$.}
\KwOut{A fractional committee~$\vec{p} = (\Delta_c)_{c \in C}$ of size $k$.}

\BlankLine
Order $C=\{c_1,c_2,\ldots,c_m\}$ so that $|N_{c_j}|\geq |N_{c_{j+1}}|$ for all $j\in[m-1]$, breaking ties arbitrarily.\; \label{alg:gcut:order_candidates}
Let $\network$ be the network formulation of instance $\instance=(C, \mathcal{A}, k)$.\;
Compute $v=val(\network).$\;
Define $\dnetwork$, a flow network with cost function $\cost$, by modifying $\network$ as follows: 
\begin{itemize}
	\item Add a ``dummy voter'' node $i^*$ with edges and capacities
	\begin{itemize}
		\item  $\capacity(s,i^*)=k-v$
		\item $\capacity(i^*,c) = \infty$ for all $c\in C$
	\end{itemize}
	\item $\cost(c_r,t) = r$ for all $r\in [m]$
	\item $\cost(u,v)=0$ for all other arcs $(u,v)$.
\end{itemize}
Compute a min-cost max flow $f$ on $\dnetwork$. \label{alg:gcut:mcmf}
\BlankLine
Let $\vec{p}$ be the fractional committee given by $f$.\; \label{alg:gcut:mf_final}
\BlankLine
\Return{$\vec{p}$}
\end{algorithm}

Due to the impossibility, we are forced to compromise when designing probabilistic committee voting rules.
In this section, we will propose a compromise by considering \emph{excludable strategyproofness} (see \Cref{sec:prelim}).
We will introduce a new rule, \emph{Generalized CUT}, which maximizes utilitarian welfare subject to \GRP (and is thus \GRP-efficient) and satisfies excludable strategyproofness.

At a high level, \GCUT computes a max flow and completes to a fractional committee of size $k$, while prioritizing candidates in order of their contribution to social welfare.
To do so, it reduces the instance to a minimum-cost maximum-flow problem on $\dnetwork$, a modified version of $\network$, which we refer to as the \emph{dummy network}.
We build the dummy network by first adding costs to arcs between candidates and the sink node $t$. 
What is important here is that the costs are \emph{distinct} and that a candidate $c$ costs less than a candidate $d$ only if $c$ is approved by at least as many voters as $d$.
Secondly, to ensure we are left with a fractional committee of size $k$, we add a new node $i^*$ to $\dnetwork$ in the voter layer.
This "dummy voter" is connected to the source and each of the candidate nodes with capacities $\capacity(s,i^*) = k-val(\network)$.
The rule then computes a minimum-cost maximum-flow solution $f$, and returns the fractional committee given by $f$.\footnote{This step can use any polynomial time method for solving minimum-cost maximum-flow problems, e.g. Cut Canceling \citep{GoTa89a,ErMc93a}. Since $val(\network)$ can be computed using any max flow algorithm, this guarantees that Generalized CUT runs in polynomial time.}
Refer to \Cref{alg:gcut} for pseudocode.

We will first make some notational remarks and formalizations. 
We assume $C=\{c_1,c_2,\ldots,c_m\}$ is ordered by decreasing approval score as is done in \GCUT.
We will use $\dnetwork_\instance$ to denote the dummy network constructed by \GCUT for the instance $\instance$.
Note that, for any instance $\instance$, $\dnetwork_\instance$ and $\network_\instance$ are both members of the class of entitlement networks (see \Cref{sec:prelim}), which we will use henceforth to state results that apply to both dummy networks and network formulations of instances.
To prove our auxiliary results with sufficient generality, we will also define a class of cost functions which is broader than those constructed by \GCUT.
We say that a cost function $\cost$ on an entitlement network $\wnetwork$ is \emph{sink-distinct} if only arcs entering the sink have non-zero costs and  all are distinct, i.e., $\cost(c,t) \neq \cost(c',t)$ for all $c,c'\in \{v: (v,t)\in \wnetwork\}$ and $\cost(u,v) = 0$ for all arcs $(u,v)\in \wnetwork$ where $v\neq t$.

Note that the cost functions constructed by \GCUT are always sink-distinct.
Since we only consider sink-distinct cost functions on entitlement networks, the only non-zero arc costs we consider are those on arcs entering the sink, which can always be interpreted as corresponding to candidates.
Thus, with slight abuse of notation, we refer to the cost of a candidate $c \in C$ for cost function $\cost$ as $\cost(c) = \cost(c, t)$.
Lastly, for an arbitrary cost function $\cost$ and candidate $c\in C$, we denote by $C(c, \cost)=\{b\in C: \cost(b) \leq \cost(c)\}$ the set of candidates which cost weakly less than $c$ under cost function $\cost.$ 
We denote the cost of a flow $f$ on a network $\wnetwork$ with cost function $\cost$ as $\cost(f)= \sum_{(u,v)\in \wnetwork} \cost(u,v) \cdot f(u,v) $. 
 Recall that the residual graph $\mathcal{R}_f$ of a flow $f$ for the min-cost flow problem has the same capacities as the residual network for the max flow problem. However, the costs on the arcs are $\cost(u, v)$ for forward arcs and $-\cost(u, v)$ for backward arcs.

\subsection{Structural Observations on \GCUT}\label{sec:Struct}

We will first recall some known facts about circulations on networks, which we will use in the lemmas to follow.
Considering the circulation $\psi=f'-f$ resulting from the subtraction of two feasible network flows with $val(f')=val(f)$, each of the following properties hold:
\begin{enumerate}[start=1,label={(\bfseries P\arabic*):}]
		\item $\psi$ is feasible on $\mathcal{R}_f$.
		\item $\psi$ is  a collection of (not necessarily disjoint) cycles $\{Q_j\}_{j\in[\ell]}$. This follows from the fact that $\psi$ is a circulation, and by applying the flow decomposition theorem for circulations (see, e.g, \cite{Erickson23}).  That is, $\psi= \sum_{j\in[\ell]} Q_j$. 
		\item The circulation resulting from the removal of any subset of cycles from $\psi $, i.e., for any $T\subseteq \{1,...,\ell\}$, the circulation $\psi'=\sum_{t\in T} Q_t$,  remains feasible on  $\mathcal{R}_{f}$. This holds because  a directed arc $(u,v)$ appears in at least one of the cycles  if and only if $\psi(u,v)>0$. Thus, for any forward arc $(u,v)$, we have $\psi'(u,v) \leq \psi(u,v) \leq \capacity(u,v) -f (u,v)$;  and for any backward arc $(v,u)$, we have $\psi'(v,u)\leq  \psi(v,u)\leq f (u,v)$. Therefore, $\psi'$ satisfies the capacity constraints on $\mathcal{R}_{f}$. 
\end{enumerate}

We now show that sink-distinct costs are sufficient to ensure that the min-cost max flow solutions of any entitlement network correspond to a unique fractional committee. To aid our later analysis, we present a more general version of this result.
\begin{lemma}\label{lem:uniq} 
	For each $\alpha\in[0,k]$, if $f$ and $f'$ are min-cost flows among all flows with value $\alpha= val(f)=val(f')$ on an entitlement network $\wnetwork$ with a sink-distinct cost function, then $f(c, t)=f'(c, t)$ for each $c\in C$.
\end{lemma}

\begin{proof}
	Let $\cost$ denote the cost function.
Consider a circulation $\psi=f'-f$, defined on the residual network $\mathcal{R}_{f}$ as follows: 
		\begin{itemize}
	\item If $f'(u,v) - f (u,v) \geq 0 $, then set  $\psi(u,v)=f'(u,v) - f (u,v)$
	\item If $f'(u,v) - f (u,v) < 0 $, then set  $ \psi(v,u)= f(u,v)-f'(u,v)$.
\end{itemize}
We note that $\psi$ is a circulation because $\alpha = val(f) = val(f')$.
By property (P2), $\psi$ can be decomposed into cycles. Let this decomposition be $\psi = \sum_{j \in [\ell]} Q_j$. 
	Since $\cost(f)=\cost(f')$, we have that  $\cost(\psi) = \sum_{j \in [\ell]} \cost(Q_j)=0 $.  
By the structure of entitlement networks, we know that each cycle in $\mathcal{R}_f$ which contains the sink node $t$, must also contain exactly two candidate nodes.
	Moreover, since cost function $\cost$ is sink-distinct, we have that $\cost(c)\neq \cost(c')$ for $c,c'\in C, c\neq c'$.
	As a result, we see that any cycle containing the sink must have non-zero cost. We now claim that each cycle in the decomposition of $\psi$ does not contain the sink node. Suppose, for a contradiction, that there exists a cycle that contains the sink node. It follows that there must exist a cycle  $Q_r$  such that $\cost(Q_r)>0$, since $\sum_{j \in [\ell]} \cost(Q_j)=0$ . By (P3), we see that $\psi'=\sum_{j\in[\ell]\setminus r} Q_j$ is feasible on $\mathcal{R}_f$. 
	Let $f''$ be the flow obtained by applying $\psi'$ on $\mathcal{R}_f$. Observe that $\cost(f'') = \cost(\psi') + \cost(f) < \cost(\psi) + \cost(f) = \cost(f')$. Since $\psi'$ is a circulation, we have $val(f'') = \alpha$. However, this contradicts the assumption that $f'$ is  min-cost flow among all flows with value $\alpha$. 
	Therefore, no cycle in the decomposition of $\psi$ contains the sink node, i.e., $\psi(c, t) = 0$ for each $c \in C$. Hence, $f'(c, t) = f(c, t)$ for each $c \in C$.
\end{proof}

\begin{remark}
	Setting $\alpha=k$ in Lemma~\ref{lem:uniq} demonstrates that the min-cost max flow solutions on any dummy network $\dnetwork$ indeed correspond to a unique fractional committee.
\end{remark}

The following lemma shows that a min-cost max flow is also a min-cost max flow on the same network restricted to the $\ell$ least costly candidates, i.e. $C(c_\ell,\cost)$, for any $\ell\in[m]$.
To be precise, for an arbitrary flow $f$ on entitlement network $\wnetwork$ and an arbitrary subset of candidates $T\subseteq C$, we define \emph{$f$ restricted to $\wnetwork(T)$} as the flow $f'$ on $\wnetwork(T)$ where $f'(i,c) = f(i,c)$ for all $i\in N\cup \{i^*\}, c\in T$ and $f'(u,v)$ is set using conservation of flows for all other arcs $(u,v).$

\begin{lemma} 
	\label{lem:top-j}
	Let $f$ be a min-cost max flow on entitlement network $\wnetwork$ with sink-distinct cost function $\cost$. 
	Then, for all $\ell\in[m]$, letting $T_\ell$ denote the subset of $\ell$ cheapest candidates under $\cost$, it holds that $f$ restricted to $\wnetwork^\ell = \wnetwork(T_\ell)$ is a min-cost max flow on $\wnetwork^\ell$.
\end{lemma}
\begin{proof}
	Fix $\ell\in[m].$ 
	Let $f_\ell$ denote $f$ restricted to $\wnetwork^\ell$.
	By construction, $f_\ell$ is a feasible flow on $\wnetwork^\ell$.
	We first show that $f_\ell$ is a max flow on $\wnetwork^\ell$.
	Suppose for a contradiction that there is an augmenting path $P=(s,i_1,d_1,i_2,d_2,\ldots,i_h,d_h,t)$ on the residual network of $f_\ell$.
	Note that (1) all arcs in $P$ are also contained in $\wnetwork$ and have equal capacities under $\wnetwork$ and $\wnetwork^\ell$ and (2) $f_\ell(u, v) = f(u, v)$ for all arcs $(u,v)$ in $P$ where $(u,v) \neq (s,i_1)$. 
	These observations follow from the construction of $f_\ell.$
	Now, observe that if $f(s, i_1) < \capacity(s, i_1)$, then $P$ is an augmenting path on $\mathcal{R}_f$ and this contradicts that $f$ is a max flow on $\wnetwork$. 
	Thus, we can assume that $f(s, i_1) = \capacity(s, i_1).$ 
	By virtue of $P$ being an augmenting path on $\mathcal{R}_{f_\ell}$, we know that $f_\ell(s, i_1) < \capacity(s, i_1) = f(s, i_1)$, and we can thus conclude by conservation of flows that $f(i_1, c)>0$ for some $c\not\in T_\ell$.
	
	Now consider the cycle $Q=(t,c_{\ell'},i_1,d_1,i_2,d_2,\ldots,i_h,d_h,t)$ in the residual network $\mathcal{R}_{f}$. 
	It is clear that $Q$ has strictly positive residual capacity.
	Furthermore, because $d_h\in T_\ell$, every candidate in $T_\ell$ is cheaper than every candidate in $C\setminus T_\ell$, and the cost function $\cost$ is sink-distinct, we know that $\cost(c) > \cost(d_h),$ and thus $Q$ is a negative cost cycle.
	However, this contradicts that $f$ is the optimal solution to the minimum-cost maximum-flow problem.

	We now show that $f_\ell$ is in fact a min-cost max flow on $\wnetwork^\ell$.
	Suppose for a contradiction that it is not. 
	Then, since $f_\ell$ is a max flow on $\wnetwork^\ell$, there must be a negative cost cycle $Q$ on the residual network of $f_\ell$, i.e., on $\mathcal{R}_{f_\ell}$.
	Since $\cost$ is sink-distinct, only arcs connected to the sink have non-zero costs, $Q$ must include the sink, $t$.
	Given this, the cycle must not include the source $s$, since otherwise it implies the existence of an augmenting path on $\mathcal{R}_{f_\ell}$.
	Thus, the flows on every arc in $Q$ are equal under $\wnetwork$ and $\wnetwork^\ell$.
	It is then immediate that $Q$ is also a negative cost cycle on $\mathcal{R}_f$, contradicting that $f$ is a min-cost max flow on $\wnetwork$.
\end{proof}

\subsection{\GRP-efficiency of \GCUT}
Using the lemmata of the previous subsection, we will now show that \GCUT maximizes utilitarian welfare subject to \GRP, and is thus \GRP-efficient.

\begin{proposition}
\label{prop:gcut}
	\GCUT maximizes utilitarian welfare subject to \GRP.
\end{proposition}
\begin{proof}
	Let $\vec{p}$ be the fractional committee returned by \GCUT for some instance $\instance$.
	By definition, $\vec{p}$ is the fractional committee given by a max flow $f$ on $\dnetwork=\dnetwork_\instance$.
	It is straightforward to see that $val(f) = k$, and thus that $\sum_{i\in N} f(s, i) = val(f) - f(s, i^*) \geq k - (k-val(\network)) = val(\network).$
	Hence, the voters' flows in $f$ (excluding the dummy voter) constitute a max flow on $\network$ and it follows that $\vec{p}$ elementwise dominates a max flow on $\network$.
	It follows from \Cref{thm:gr_characterization} that $\vec{p}$ satisfies \GRP.
	
	We must now show that $\vec{p}$ maximizes utilitarian welfare amongst \GRP outcomes.
	Let $\vec{q}$ denote some fractional committee which satisfies \GRP. 
	Our goal is to show that the social welfare attained by $\vec{p}$ is at least that of $\vec{q}$, i.e. $\sum_{i\in N} u_i(\vec{p}) \geq \sum_{i\in N} u_i(\vec{q})$.
	We first note that $\vec{q}$ is the fractional committee given by some max flow $f'$ on $\dnetwork.$
	To see this, recall from the characterization of \GRP (\Cref{thm:gr_characterization}) that $\vec{q}$ must elementwise dominate some max flow $g$ on $\network.$
	Then, consider the flow $f'$ on $\dnetwork$ constructed by setting $f'(i^*, c) = q_c - g(c, t)$ for all $c\in C$ and setting $f'(i, c) = g(i, c)$ for all $i\in N, c\in C$, and setting all other flows in accordance with conservation of flows.
	It can be verified that $f'$ constitutes a max flow on $\dnetwork$ (also noting that $\sum_{c\in C} q_c = k$).

	We will first show that $\vec{p}$ stochastically dominates $\vec{q}$.

	\begin{restatable}{claim}{claimsd} \label{claim:sd}
		$\sum_{r\in [\ell]} p_r \geq \sum_{r\in [\ell]} q_r$ for all $\ell\in [m]$.
	\end{restatable}
	\begin{proof}[Proof of Claim]\
		Fix an arbitrary $\ell\in [m].$
		Consider the network $\dnetwork^\ell = \dnetwork(\bigcup_{r\in[\ell]} c_r)$.
		Let $f_\ell$ be the flow $f$ restricted to $\dnetwork^\ell$.
										Since $\dnetwork$ is an entitlement network, $f$ is a min-cost max flow on $\dnetwork$ with sink-distinct cost function, and $\bigcup_{r\in[\ell]} c_r$ are the $\ell$ cheapest candidates under that cost function, we know from \Cref{lem:top-j} that $f_\ell$ is a max flow on $\dnetwork^\ell$.
		This leads to the statement of our claim:
		\begin{align*}
			\sum_{r\in [\ell]} p_r &= \sum_{r\in [\ell]} f(c_r, t) \\
			&\geq \sum_{r\in [\ell]} f'(c_r, t) =  \sum_{r\in [\ell]} q_r.
		\end{align*}
		where the inequality follows since $f'$ restricted to $\dnetwork^\ell$ must not have a greater value than $f_\ell$, which is a max flow on $\dnetwork^\ell$.
	\end{proof}

	Let $\indicator{\cdot}$ be an indicator function and let $a_j$ denote the approval score of each candidate in $c_j$, i.e. $a_j=|N_{c_j}|, \ \forall j\in[m]$. 
	We now have what we need to show that $\vec{p}$ attains a greater social welfare than $\vec{q}$:
	\begin{align*}
		\sum_{i\in N} u_i(\vec{p}) - \sum_{i\in N} u_i(\vec{q}) &= \sum_{i\in N} \sum_{c\in C} p_c \indicator{c\in A_i} - \sum_{i\in N} \sum_{c\in C} q_c \indicator{c\in A_i} \\
		&= \sum_{c\in C} |N_c| \cdot (p_c - q_c) \ = \ \sum_{j\in [m]} a_j \cdot (p_{c_j} - q_{c_j}) \geq 0
	\end{align*}
	The final inequality follows from \Cref{lem:sequence_product} once we note that $\{a_j\}_{j\in[m]}$ is non-increasing and thatthe sequence $\{p_j - q_j\}_{j\in[m]}$ has non-negative partial sums by \Cref{claim:sd}.
\end{proof}
\begin{restatable}{lemma}{lemsequence}
	\label{lem:sequence_product}
	Let $\{a_i\}$ and $\{b_i\}$ be two sequences of $r$ real numbers such that $a_1\geq a_2 \geq \ldots \geq a_r\geq 0$ and for all $r'\in [r]$, it holds that $\sum_{i\in [r']} b_i\geq 0$.
	Then, it must be that $\sum_{i\in [r]} a_i\cdot b_i \geq 0$. 
\end{restatable}
\begin{proof}
	Note that for each $j\in [r]$, we can write $b_j = \sum_{i\in [j]} b_i -  \sum_{i\in [j-1]} b_i$. Substituting, we can then transform our expression of interest as follows:
	\begin{align*}
		&\sum_{i\in [r]} a_i \cdot b_i\\
		&= a_1 \sum_{i=1}^{1} b_i + (a_2\sum_{i=1}^2 b_i - a_2 \sum_{i=1}^1 b_i) + \ldots + (a_{r-1} \sum_{i=1}^{r-1} b_i - a_{r-1} \sum_{i=1}^{r-2}b_i) + (a_r \sum_{i=1}^{r} b_i - a_r \sum_{i=1}^{r-1}b_i)\\
		&\geq a_2 \sum_{i=1}^{1} b_i + (a_3\sum_{i=1}^2 b_i - a_2 \sum_{i=1}^1 b_i) + \ldots + (a_{r} \sum_{i=1}^{r-1} b_i - a_{r-1} \sum_{i=1}^{r-2}b_i) + (a_r \sum_{i=1}^{r} b_i - a_r \sum_{i=1}^{r-1}b_i)\\
		&\geq a_r \sum_{i=1}^{r} b_i \geq 0.
	\end{align*}
	The first inequality uses the fact that each of the partial sums of $\{b_i\}$ are non-negative and the fact that $\{a_i\}$ are non-increasing to replace $a_j\cdot \sum_{i=1}^j b_i$ with the weakly smaller quantity $a_{j+1}\cdot \sum_{i=1}^j b_i$ for all $j\in[r-1]$. The second inequality is due to cancellation. The third inequality again follows from the non-negativity of $\{a_i\}$ and of $\sum_{i\in [r]} b_i$. 
\end{proof}

We remark that, using the same argument used to prove \Cref{prop:gcut}, it can be shown that a simple modification of \GCUT maximizes weighted utilitarian welfare subject to \GRP for arbitrary weights.
Furthermore, it is our view that \GCUT nicely demonstrates how network flow tools can be used by voting rules to optimize an objective subject to GRP more broadly.

\subsection{Excludable Strategyproofness of \GCUT}
In this subsection, we will prove the following theorem, which states that \GCUT is excludable strategyproof.

\begin{theorem} \label{thm:excl_sp}
	\GCUT satisfies excludable strategyproofness.
\end{theorem}

To show that \GCUT is excludable strategyproof, we will show that the rule (1) does not allow manipulations in which a voter reports a superset of their true approval set, and (2) does not allow subset manipulations when a voter can only benefit from candidates they say they approve.
To show that superset manipulations are not tenable, the critical step, captured by our next lemma, is to show a sort of monotonicity condition: between two similar enough networks, the flow into a candidate whose cost increases should not increase.

\begin{lemma} \label{lem:monotonicity}
	Suppose $\wnetwork$ and $\wnetwork'$ are two entitlement networks that differ only in that $\wnetwork'$ may include additional arcs of the form $(u,v)$ for any $u\in N,v\in C$.
	Let $f$ and $f'$ be min-cost max flow solutions to $\wnetwork$ and $\wnetwork'$ with sink-distinct cost functions $\cost$ and $\cost'$, respectively.
	Furthermore, let $c\in C$ and $\wnetwork_0$ and $\wnetwork'_0$ be the networks $\wnetwork$ and $\wnetwork'$ restricted to candidates in $C(c, \cost')$. 
	If it holds that $C(c, \cost)\subseteq C(c, \cost')$,
	then at least one of the following must hold:
	\begin{enumerate}[label=(\arabic*)]
		\item $f(c, t)\geq f'(c, t)$ \label{enum:then1}
		\item $val(\wnetwork_0') > val(\wnetwork_0)$. \label{enum:then2}
	\end{enumerate}
\end{lemma}
\begin{proof}
	Suppose $f'(c, t) > f(c, t)$, i.e., statement \ref{enum:then1} does not hold. 
	We will show that statement \ref{enum:then2} must then hold, i.e., $val(\wnetwork_0') > val(\wnetwork_0)$.
	Let $f_0$ and $f'_0$ denote min-cost max flow solutions to the networks $\wnetwork_0$ and $\wnetwork'_0$ with respective cost functions $\cost$ and $\cost'$.
	Note that, by definition, $\wnetwork'_0 = \wnetwork'(C(c,\cost')).$
	Thus, \Cref{lem:top-j} tells us that $f'$ restricted to $\wnetwork'_0$ is a min-cost max flow solution to $\wnetwork'_0.$
	By \Cref{lem:uniq}, we can thus conclude that $f'_0$ and $f'$ have identical flows entering the sink, i.e. $f'_0(d , t) = f'(d, t)$ for all $d\in C(c,\cost')$. 
	
	Now, using the condition from the lemma statement, we know that every candidate in $C(c,\cost)$ is included in $C(c,\cost')$ and thus is also included in $\wnetwork_0$.
	Again, using \Cref{lem:top-j}, we see that if we restrict $f_0$ to arcs on $\wnetwork_0(C(c,\cost))$, we obtain a min-cost max flow solution on $\wnetwork_0(C(c,\cost)).$
	By a parallel argument, the same is true of flow $f$, when restricted to arcs on $\wnetwork_0(C(c,\cost))$.
	We can thus conclude that $f_0(d, t) = f(d, t)$ for all $d\in C(c, \cost)$.
	Taken together, this means that $f'_0(c, t) = f'(c, t) > f(c, t) = f_0(c, t).$

	Observe that $f'_0$ is a max flow on $\wnetwork'_0$, and thus by \Cref{lem:top-j}, $f'_0$ restricted to $\wnetwork'_0(C(c, \cost')\setminus \{c\})$ gives a max flow on that subnetwork.
	Recalling that $\wnetwork'$ is simply $\wnetwork$ with extra arcs, we see that $f_0$ is a feasible flow on $\wnetwork'_0$, and thus, that $f_0$ restricted to $\wnetwork'_0(C(c, \cost')\setminus \{c\})$ is also feasible on that subnetwork.
	This means that $C(c, \cost')\setminus \{c\}$ must induce a weakly greater flow under $f'_0$ than under $f_0$, i.e., 
		\begin{align*}
			\sum_{\crampedclap{d \in C(c, \cost')\setminus \{c\}}} f'_0(d, t) - f_0(d, t) \geq 0 &\implies \\	
			val(\wnetwork'_0) - f'_0(c , t) \geq val(\wnetwork_0) - f_0(c , t) &\implies \\
			val(\wnetwork'_0) > val(\wnetwork_0).
		\end{align*}
\end{proof}

We can now show that \GCUT does not allow superset manipulations.
We denote the voting rule function counterpart of \GCUT as $GCUT$.
\begin{proposition}	\label{prop:sssp}
	Suppose $\mathcal{A} = (A_1', \ldots, A_i, A_{i+1}', \ldots, A_n')$ and $\mathcal{A}' = (A_1', \ldots, A_i', A_{i+1}', \ldots, A_n')$ where $A_i\subseteq A_i'$.
	Then, for each $k\leq m$, it holds that $u_i(GCUT(\mathcal{A},k))\geq u_i(GCUT(\mathcal{A}',k))$.
\end{proposition}
\begin{proof}
	Fix $k\leq m$ and let $\vec{p}=GCUT(\mathcal{A},k)$ and $\vec{q}=GCUT(\mathcal{A}',k)$ denote the fractional committees computed by \GCUT under $\mathcal{A}$ and $\mathcal{A}'$, respectively.
	Let $\network=\network_\instance$ and $\network'=\network_{\instance'}$ denote the respective network formulations of the instances $\instance=(C, \mathcal{A}, k)$ and $\mathcal{I'}=(C, \mathcal{A'}, k)$.
	Note that, since $\mathcal{A}'$ is the approval profile resulting from $\mathcal{A}$ with some additional approvals, it is clear that any feasible flow on $\network$ is also feasible on $\network'$ and thus that $val(\network)\leq val(\network').$
	We will proceed by cases that distinguish whether this inequality is strict.

	\smallskip
	\noindent\textbf{\textit{Case I: $val(\network) < val(\network')$.} }

	First, since \GCUT satisfies \GRP, we know from \Cref{thm:gr_characterization} that there exist max flows $g$ on $\network$ and $g'$ on $\network'$ such that $p_c\geq g(c, t)$ and $q_c\geq g'(c, t)$ for all $c\in C.$
	Since $g$ is feasible on $\network'$, it must be that it induces an augmenting path on the residual network of $\network'$.
	Let $P=(s,i_1,d_1,\ldots,i_h,d_h,t)$ denote some such augmenting path.
	Since $g$ is a max flow on $\network$, there cannot be an augmenting path on $\mathcal{R}_g$.
	Thus, $P$ must contain one of the new arcs in $\network'$, and thus $i_r = i$ for some $r\in [h]$ (and $d_r\in A_i'\setminus A_i$).
	Observe that $\tilde P=(s,i_1,d_1,\ldots,i,c,t)$ is an augmenting path in $\mathcal{R}_g$ for any $c\in A_i$ with $g(c, t)<\capacity(c, t) = 1$.
	This holds because $\tilde P$ only uses arcs present in $\network$.
	Therefore, $g(c, t)=1$ for all $c\in A_i$.
	Altogether, this leads us to our statement: 
	$$u_i(GCUT(\mathcal{A},k)) = \sum_{c\in A_i} p_c \geq \sum_{c\in A_i} g(c, t) = |A_i| \geq u_i(GCUT(\mathcal{A}',k)).$$

	\smallskip
	\noindent\textbf{\textit{Case II: $val(\network) = val(\network')$.} }

	Let $\dnetwork=\dnetwork_\instance$ and $\dnetwork'=\dnetwork_{\instance'}$ denote the respective dummy networks constructed by \GCUT under instances $\mathcal I$ and $\mathcal I'$, respectively.
	Let $\cost$ and $\cost'$ denote the arc cost functions associated with $\dnetwork$ and $\dnetwork'$, respectively.
	Let $\capacity$ and $\capacity'$ denote the arc capacities of $\dnetwork$ and $\dnetwork'$, respectively.
	Let $f$ and $f'$ be min-cost max flow solutions to $\dnetwork$ and $\dnetwork'$, respectively.
	Observe by \Cref{lem:uniq} that $\vec{p}$ and $\vec{q}$ are the fractional committees given by $f$ and $f'$, respectively.
	Note that, because $val(\network) = val(\network')$, it holds that $\capacity(s, i^*) = \capacity'(s, i^*)$.
	It can thus be observed that $\dnetwork'$ is the same as $\dnetwork$ with the exception of additional arcs between the voter and candidate node layers.
	
	Now consider an arbitrary candidate $a\in A_i$ and let $\dnetwork_0$ and $\dnetwork'_0$ be the networks $\dnetwork$ and $\dnetwork'$ restricted to candidates in $C(a, \cost')$, i.e., $\dnetwork_0=\dnetwork(C(a,\cost'))$ and $\dnetwork'_0=\dnetwork'(C(a,\cost'))$. 
	Note that for every candidate $d$ for which $\cost(d)<\cost(a)$, it holds that $\cost'(d) < \cost'(a)$.
	This is true because the number of voters approving $d$ under $\mathcal A'$ is weakly greater than under $\mathcal A$ whereas the number of voters approving $a$ is equal in both approval profiles. Thus, if $|N_d|\geq |N_a|$ under $\mathcal A$, then this also holds under $\mathcal A'.$
	In summary, we can conclude that $C(c,\cost)\subseteq C(c, \cost')$, i.e., the conditions of \Cref{lem:monotonicity} hold for entitlement networks $\dnetwork$ and $\dnetwork'$ and candidate $a$.

	Thus, from \Cref{lem:monotonicity}, we can conclude that either $f(a, t)\geq f'(a, t)$ or $val(\dnetwork_0') > val(\dnetwork_0)$.
	Consider the latter case.
	As we have already argued in Case I, the candidates approved by voter $i$ in $\dnetwork_0$ must be fully saturated in this case.
	In particular, it must be that $f(a, t) = 1 \geq f'(a, t)$.
	Thus, we can conclude that $f(a, t)\geq f'(a, t)$ in either case.
	Since we have proved this with sufficient generality, this statement holds for all $a\in A_i.$
	This leads us to the statement of the proposition:
	$$u_i(GCUT(\mathcal{A})) = \sum_{c\in A_i} p_c = \sum_{c\in A_i} f(c, t) \geq \sum_{c\in A_i} f'(c, t) = \sum_{c\in A_i} q_c = u_i(GCUT(\mathcal{A}')).$$
\end{proof}
It remains to show that \GCUT does not allow subset manipulations when a voter cannot benefit from candidates they do not approve. 
We first prove a technical lemma that applies to entitlement networks and, consequently, to both $\network$ and $\dnetwork$.

	\begin{lemma}\label{lem:networkMon}
		For a fixed voter $i\in N$, let  $\wnetwork$ and $\wnetwork'$ be two entitlement networks that differ only in that $\wnetwork$ optionally contains extra arcs  of the form $(i,c)$ for some $c\in C$.
		For a fixed  sink-distinct cost function $\cost$ defined on both networks, let $g$ and $g'$ represent  min-cost max flows on $\wnetwork$ and $\wnetwork'$, respectively. Then,
		$$\sum_{c\in C : (i,c)\in \wnetwork'} g'(c, t) \leq \sum_{c\in C : (i,c)\in \wnetwork} g(c, t). $$
	\end{lemma}
	\begin{proof}
		Recall that each entitlement network corresponds to a unique instance $\instance = (C, \mathcal{A}, k)$. 
		Thus, the set $\{c \in C : (i, c) \in \wnetwork'\}$ represents the candidates approved by agent $i$ in the instance corresponding to $\wnetwork'$, which we denote as $A_i'$. 
		Similarly, we define $A_i$ as $\{c \in C : (i, c) \in \wnetwork\}$. 
		By the assumption of the lemma, we have $A_i' \subseteq A_i$.  
		We denote $A_i^- := A_i \setminus A_i'$. 
		
		Consider the flow $g^{0}=g'$ on the network $\wnetwork$. Since $\wnetwork'$ is a subnetwork of $ \wnetwork$, we have that $g^{0}$ is feasible on $ \wnetwork$.  Let $\tildG$ be the flow resulting from the removal of all cost-cancelling cycles on  the residual network $\mathcal{R}_{g^{0}}$.   
		As $\mathcal{R}_{\tildG}$ does not contain cost-cancelling cycles, we know that we can compute a min-cost max flow on $\wnetwork$ by applying a sequence of cheapest augmenting paths until no augmenting paths remain (see e.g., \citep{FlowBook93}).
		We call this flow $h$. 
		By \Cref{lem:uniq}, we have that $h(c, t)=g(c, t)$ for each $c\in C$. 
		Since $h$ is obtained by a sequence of cheapest augmenting paths on $\mathcal{R}_{\tildG }$, the flow from each candidate to the sink  does not decrease. It follows that, 
		$$
		\sum\limits_{c\in A_i} \tildG( c , t ) \leq  	\sum\limits_{c\in A_i} h( c , t ) = \sum\limits_{c\in A_i} g( c , t ). 
		$$
		Therefore,  we see that it suffices to show $\sum\limits_{c\in A'_i} g^0( c , t) \leq  \sum\limits_{c\in A_i} \tildG ( c , t)$, since $g'( c , t) =  g^0 ( c , t)$ for each $c\in C$.  
		
		For each arc $(u,v)$,  we  define an auxiliary flow $\psi$ on $\mathcal{R}_{g^{0}}$ as follows:
		\begin{itemize}
			\item If $\tildG(u,v) - g^0 (u,v) \geq 0 $, then set  $\psi(u,v)=\tildG(u,v) - g^0 (u,v)$
			\item If $\tildG(u,v) - g^0 (u,v) < 0 $, then assign flow on the reverse arc $(v,u)$ and set $ \psi(v,u)= g^0 (u,v)-\tildG(u,v)$.
		\end{itemize}
		Since $val(g^0)=val(\tildG)$ by construction, and by conservation of flows, we have that  $\psi$ is, in fact, a circulation on  $\mathcal{R}_{g^{0}}$. Thus, $\psi$ and its cycle decomposition $\{Q_j\}_{j \in [\ell]}$ adhere to properties (P1), (P2), (P3), outlined in \Cref{sec:Struct}. 

We first show that removing every \textit{non-negative} cost cycle from $\psi$ produces a circulation on $\mathcal{R}_{g^{0}}$, which, when applied to $\mathcal{R}_{g^{0}}$, results in a  flow that is identical to $\widetilde{g}$ on all arcs entering the sink. 
		Let $\widetilde{\psi}$ be a circulation defined by removing all cycles $Q_j$ with non-negative cost i.e., $\widetilde{\psi}= \sum_{ j\in H } Q_j$ where $H:=\{ j\in [\ell] \ : \  \cost( Q_j)< 0   \}$ . 
		By (P3), we know that $\widetilde{\psi}$ is feasible on $\mathcal{R}_{g^{0}}$.  
		Furthermore, after applying the circulation $\widetilde{\psi}$ to $\mathcal{R}_{g^{0}}$, the resulting flow $\widetilde{h}$ on $\wnetwork$ has cost 
		\begin{align*}
			\cost(\widetilde{h}) = \cost(g^0) + \cost(\widetilde{\psi}) \leq \cost(g^0) + \cost(\psi) = \cost(\widetilde{g}).
		\end{align*}
		Also note that $val(\widetilde{h}) = val(g^0)$.
		Since $\widetilde{g}$ is the min-cost flow among all flows on $\wnetwork$ with value $val(g^0)$, we conclude that $\widetilde{h}$ must also be a min-cost flow among all such flows on $\wnetwork$.
		Finally, by Lemma~\ref{lem:uniq}, we see that $\widetilde{h}(c , t) =\tildG (c , t)$ for all $c\in C$. It follows that,   $\widetilde\psi$ satisfies $\widetilde\psi(c, t) = \widetilde{g}(c , t)-g^0(c , t)$ for all $c\in C$.
		
		Consider the cycles in the decomposition of $\widetilde\psi$. 
			Note that each cycle in $\{Q_j\}_{j\in H}$ must include an arc of the form $(i, a)$ for some $a \in A_i^-$. 
		To see this, suppose for a contradiction, that there exists a  negative cost cycle $Q_r$ which does not contain an arc of the form $(i, a)$ for some $a \in A_i^-$. By (P3), $Q_r$ is feasible on $\mathcal{R}_{g^{0}}$.
		Recall that $\wnetwork$ and $\wnetwork'$ differ only in that $\wnetwork$ contains extra arcs of the form $\{ (i,c) \ : \ c\in A^-_i \}$.  Also, since $g^{0}(u,v)=g'(u,v)$ for all $(u,v)\in \wnetwork'$, we see that $Q_r$ must also be feasible on the residual network $\mathcal{R}_{g'}$ of flow $g'$ on $\wnetwork'$. This contradicts the assumption that $g'$ is min-cost max flow on $\wnetwork'$.

		Letting $L:= \{ c\in A'_i \ : \ \tildG (c , t) \geq g^0(c, t) \}$, we see that 
		$$\sum\limits_{c\in A'_i} \left(  \tildG (c , t) -  g^0 (c , t) \right)= \sum\limits_{c\in A'_i \setminus L }- \widetilde\psi(t , c) +\sum\limits_{c\in L} \widetilde\psi(c, t). $$
		
		Finally,  observe that 
		\begin{align*} 
			\sum\limits_{c\in A_i}  \tildG (c , t) - \sum\limits_{c\in A'_i} g^0 ( c , t) &=
			\sum\limits_{c\in A'_i} \left(  \tildG (c , t) -  g^0 (c , t) \right) + \sum\limits_{c\in A^{-}_i}   \tildG (c , t) \\ 
			&= \sum\limits_{c\in A'_i \setminus L }- \widetilde\psi(t , c) +\sum\limits_{c\in L} \widetilde\psi(c, t) + \sum\limits_{c\in A^{-}_i}   \tildG (c , t) \\
			& \geq \sum\limits_{c\in A'_i \setminus L }- \widetilde\psi(t , c) + \sum\limits_{c\in A^{-}_i}   \tildG (c , t) \\ 
			&= \sum\limits_{c\in A'_i \setminus L }\ \ \sum_{j\in H: (t, c)\in Q_j \ }- Q_j(t, c) + \sum\limits_{c\in A^{-}_i}   \tildG (c , t) \\ 
			&\geq \sum\limits_{c\in A'_i \setminus L }\ \ \sum_{j\in H: (t, c)\in Q_j \ }- Q_j(t, c) + \sum\limits_{c\in A^{-}_i}   \tildG (i , c) \\ 
			&\geq \sum\limits_{c\in A'_i \setminus L }\ \ \sum_{j\in H: (t, c)\in Q_j \ }- Q_j(t , c) + \sum\limits_{c\in A^{-}_i}\ \ \sum_{j\in H: (i,c)\in Q_j \ }Q_j(i , c)    \\ 
			& = \sum\limits_{c\in A'_i \setminus L } \ \  \sum_{j\in H: (t,c)\in Q_j \ }- Q_j(t, c) +  \sum_{j\in H } val(Q_j) \\
			& \geq \sum_{j\in H \ }- val(Q_j) +  \sum_{j\in H } val(Q_j)  \\
			&\geq 0   
		\end{align*}
		where the second to last inequality follows from the fact that each cycle contains exactly one arc of the form $(t, c)$ for some $c\in C$. Hence, the sum adds each cycle's value at most once.
	\end{proof}

The following lemma states that the dummy network flows computed by \GCUT candidate-wise dominates any min-cost max flow solution to the network formulation of the corresponding instance.
\begin{lemma}\label{lem:phaseII}
	If $g$ and $f$ are min-cost max flow solutions on $\network_\instance$ and $\dnetwork_\instance$, respectively, for some instance $\instance=(C, \mathcal{A}, k)$ and a sink-distinct cost function,  then for each $c\in C$, we have $f(c, t)\geq g(c , t)$.
\end{lemma}
\begin{proof}
	Let $\network$ and $\dnetwork$ denote $\network_\instance$ and $\dnetwork_\instance$, respectively.
	Note that $\network$ is a subnetwork of $\dnetwork$. 
	Now consider the residual network of $g$ in $\dnetwork$. 
	We claim that a min-cost max flow on $\dnetwork$ can be obtained by a sequence of cheapest augmenting paths on the residual network of $g$ on $\dnetwork$. 
	Let $h$ be a flow obtained by a sequence of cheapest augmenting paths on the residual network of $g$ on $\dnetwork$, by breaking ties in favor of the shortest path, until no augmenting paths remains i.e., $h$ is a max flow on $\dnetwork$. 
	We note that the first augmenting path must include the node $i^*$. 
	This holds because $g$ is a maximum flow on $\network$. 
	Furthermore, since $i^*$ is connected to every candidate, the augmenting path must take the form $(s , i^* , c , t)$ for some $c \in C$. 
	It follows that the flow on all arcs in $\network $ remains unchanged, i.e., the flow on each arc $(u, v) \in \network$ is exactly $g(u, v)$. 
	By repeating this argument iteratively, we observe that each subsequent augmenting path also leaves the flow values on arcs in $\network$ unchanged. 
	It follows that $h(u, v) = g(u, v)$ for each $(u, v) \in \network$. 
	
	Suppose, for a contradiction, that $h$ is not a minimum-cost maximum flow. Then there must exist a cost-canceling cycle in the residual network of $h$. Since $h(u, v) = g(u, v)$ for each $(u, v) \in \network$, the cost-canceling cycle must include the node $i^*$; otherwise, this would contradict the fact that $g$ is a min-cost max flow on $\network$.
	Since the dummy voter $i^*$ in $\dnetwork$ is saturated, the only arc between $i^*$ and $s$ available in the residual network of $h$ is $(i^*,s)$. 
	By the structure of $\dnetwork$ and  the cost function being sink-distinct, we know that every non-zero cost cycle must include the sink node $t$. 
	Thus, the cost-canceling cycle must include $t$ as well as the arc $(i^*,s)$, but this implies the existence of an s-t path not involving the dummy node $i^*$ which was available on the residual network of $g$ on $\network$, contradicting the fact that $g$ is max flow on $\network$. 
	
	Finally, as $h$ was obtained by a sequence of cheapest augmenting paths on the residual graph of $g$ on $\dnetwork$, we have that $h(c, t)\geq g(c, t) $ for each  $c\in C$. By Lemma~\ref{lem:uniq}, we see that $h(c, t)=f(c, t)$ for all $c\in C$. Thus, indeed, for each $c\in C$, we have $f(c, t) \geq g(c, t)$.
\end{proof}

With these lemmas in hand, we can now prove that \GCUT does not allow subset manipulations when excluding voters from candidates they claim not to approve.
\begin{proposition}	\label{prop:subset_exsp}
	Suppose $\mathcal{A} = (A_1', \ldots, A_i, A_{i+1}', \ldots, A_n')$ and $\mathcal{A}' = (A_1', \ldots, A_i', A_{i+1}', \ldots, A_n')$ where $A_i'\subseteq A_i$.
	Then, for each $k\leq m$, it holds that $u_i(\vec{p}) = \sum_{c\in A_i} p_c \geq \sum_{c\in A_i'} q_c$ where $\vec{p}=GCUT(\mathcal A,k)$ and $\vec{q}=GCUT(\mathcal A',k)$.
\end{proposition}
\begin{proof}
	Fix $k\leq m$ and let $\network=\network_\instance$ and $\network'=\network_{\instance'}$ denote the respective network formulations of the instances $\instance=(C, \mathcal{A}, k)$ and $\mathcal{I'}=(C, \mathcal{A'}, k)$.
	Similarly, let $(\dnetwork,\cost)$ and $(\dnetwork',\cost')$ represent the dummy networks along with the cost functions created by \GCUT under the instances $\instance$ and $\instance'$, respectively. 
	Also, let $f$ and $f'$ be min-cost max flows on the respective dummy networks with the associated costs, and note by \Cref{lem:uniq} that $p_c=f(c, t)$ and $q_c=f'(c, t)$ for all $c\in C$ . 
	Let $A^-_i \coloneqq A_i\setminus A'_i$, and  define \textit{critical arcs} as $\mathcal{H}\coloneqq \{ (i,c) \ | \ c\in A_i^{-} \}$.  
	As $\network'$ is a subnetwork of $\network$, we have that $val(N)\geq val(\network')$.  
	We divide the analysis into two based on whether $val(\network)= val(\network')$ or  $val(\network)> val(\network')$. 
	
	\smallskip
	\noindent\textbf{\textit{Case I: $val(\network)= val(\network')$.} }
	
	Note that the difference between the networks $\network$ and $\network'$ lies solely in the absence of the critical arcs $\mathcal{H}$ in $\network'$. Furthermore, based on the case distinction, the capacities entering the dummy voters are identical in both $\dnetwork'$ and $\dnetwork$. 
	Thus, $\dnetwork'$ and $\dnetwork$ are entitlement networks that differ only in that $\dnetwork$ includes the critical arcs $\mathcal{H}$.  
	Let $\widetilde{f}$ be some min-cost max flow on $\dnetwork$ with cost function  $\cost'$.
	By applying \Cref{lem:networkMon} on $\dnetwork'$ and $\dnetwork$ with cost function $\cost'$, we get that 
	\begin{equation}\label{eq:1}
		\sum\limits_{c\in A'_i} f'( c , t) \leq  \sum\limits_{c\in A_i} \widetilde{f}( c , t).  
	\end{equation}
	
	Observe that for each $c \in C\setminus A_i$, any candidate which costs less than $c$ under $\cost'$ also costs less than $c$ under $\cost$, i.e. $C(c,\cost')\subseteq C(c,\cost)$. 
	Thus, the cost functions satisfy the conditions of \Cref{lem:monotonicity}. 
	By applying \Cref{lem:monotonicity} to network $\dnetwork$ and cost functions $\cost$ and $\cost'$, we get that $\widetilde{f}(c, t) \geq f(c, t)$ for all $c\in C\setminus A_i$, since Statement~\ref{enum:then2} must be false.
	Since $val(\widetilde{f}) = val(f)$, this means that $\sum_{c\in A_i} \widetilde{f}(c, t)\leq \sum_{c\in A_i} f(c, t)$.  
	Combining with equation~(\ref{eq:1}), we obtain the desired inequality,
	$$
	\sum_{c\in A_i'} q_c = \sum_{c\in A_i'} f'(c, t) \leq 	\sum\limits_{c\in A_i} \widetilde{f}( c , t) \leq  \sum\limits_{c\in A_i} f( c , t)= \sum_{c\in A_i} p_c.  
	$$

\smallskip
\noindent\textbf{\textit{Case II: $val(\network)> val(\network')$.} }

Let $g'$ and $g$ be the min-cost max flow on $(\network', \cost')$ and $(\network,\cost)$ respectively. 
Recall that the difference between the networks $\network$ and $\network'$ lies solely in the absence of the critical arcs $\mathcal{H}$ in $\network'$.   
Denote $\widetilde{g}$ as the min-cost max flow on $\network$ with cost function  $\cost'$. 
By  applying \Cref{lem:networkMon} on $\network'$ and $\network$ with cost function $\cost'$, we obtain $\sum\limits_{c\in A'_i} g'( c , t) \leq  \sum\limits_{c\in A_i} \widetilde{g}( c , t)$.   
Observe that for each $c \in C\setminus A_i$, any candidate which costs less than $c$ under $\cost'$ also costs less than $c$ under $\cost$, i.e. $C(c,\cost')\subseteq C(c,\cost)$. 
Thus, the cost functions satisfy the conditions of \Cref{lem:monotonicity}. 
Again, by applying \Cref{lem:monotonicity} to network $\network$ and cost functions $\cost$ and $\cost'$, we get that $\widetilde{g}(c, t) \geq g(c, t)$ for all $c\in C\setminus A_i$.
Since $val(g) = val(\widetilde{g})$, it holds that $\sum_{c\in A_i} \widetilde{g}(c, t)\leq \sum_{c\in A_i} g(c, t)$.  
Combining the two inequalities obtained thus far, we get 
\begin{equation}	\label{eq:caseII}
	\sum_{c\in A'_i} g'(c, t) \leq  \sum_{c\in A_i} \widetilde{g}(c, t) \leq \sum_{c\in A_i} g(c, t).
\end{equation}

We point out that $g'(c, t)=1$ for each $c\in A'_i$, which holds due to the following reasoning. 
Since  $val(\network)> val(\network')$ and these networks only differ by the presence of critical arcs, there must be some augmenting path $P$ in the residual network of $g'$ in $\network$. 
Note that, since $g'$ is a max flow on $\network'$, $P$ must include a critical arc to avoid constituting an augmenting path on $\mathcal{R}_{g'}.$
Thus, there is a path to $i$ in the residual network of $g'$ in $\network'$.
Again, because $g'$ is a max flow on $\network'$, this implies that $g'(c, t) = 1$ for all $c\in A_i'$, since otherwise there would be an augmenting path in $\mathcal{R}_{g'}.$

We can now establish the statement of the proposition:
\begin{align*}
\sum_{c\in A_i} p_c = \sum_{c\in A_i} f(c, t) &\geq \sum_{c\in A_i} g(c, t) 
\geq \sum_{c\in A'_i} g'(c, t) = |A_i'| \geq \sum_{c\in A'_i} f'(c, t) = \sum_{c\in A_i'} q_c
\end{align*}
where the first inequality follows from the application of \Cref{lem:phaseII} to instance $\instance$ and the second inequality follows from \Cref{eq:caseII}.

\end{proof}

\begin{proof}[Proof of \Cref{thm:excl_sp}]
	Let $\mathcal{A}=(A'_1,\ldots,A_i,\ldots,A'_n)$ be an arbitrary approval profile with some truthful voter $i\in N.$
	Consider some alternative approval profile obtained from an arbitrary unilateral misreport from voter $i$, denoted $\mathcal{A}'=(A'_1,\ldots,A'_i,\ldots,A'_n)$.
	Let $\vec{p}$ and $\vec{q}$ denote the fractional committees returned by \GCUT with arbitrary input committee size $k\leq m$ and approval profiles $\mathcal{A}$ and $\mathcal{A}'$, respectively.
	Let $A_0 = A_i\cap A'_i $ and $A^+ = A'_i\setminus A_i$ and $A^- = A_i\setminus A_i'$.
	That is, $A_i = A_0\cup A^-$ and $A_i' = A_0\cup A^+.$
	To prove that \GCUT is excludable strategyproof, we would like to show that $u_i(\vec{p})\geq \sum_{c\in A_0} q_c$.

	Consider $\mathcal{A}'' = (A'_1,\ldots,A_0,\ldots,A'_n)$ and let $\vec{x}$ be the fractional committee returned by \GCUT with input committee size $k$ and approval profile $\mathcal{A}''$.
	By \Cref{prop:subset_exsp}, we can conclude that $u_i(\vec{p}) \geq \sum_{c\in A_0} x_c.$
	Furthermore, we can see that $\mathcal{A}'$ is simply $\mathcal{A}''$ with voter $i$ unilaterally misreporting a superset.
	Thus, by \Cref{prop:sssp}, we see that $\sum_{c\in A_0} x_c \geq \sum_{c\in A_0} q_c,$ yielding the desired statement.
\end{proof}

\section{Best-of-both-worlds Fairness}
\label{sec:bbw}
In this section, we will build upon the characterization of our new axiom to obtain a general \emph{best-of-both-worlds} result. 
This result leads us to strengthen two known compatibility results from \citet{ALS+23}, answering an open question left by the authors in the process.
We begin with a lemma which, at a high level, shows that any fractional committee which can be represented by a network flow can be completed to a max flow.
\begin{lemma} \label{lem:mf_domination}
	Let $\vec{q}$ be the committee given by a feasible flow on network representation $\network$. 
	There exists a max flow $f^*$ on $\network$, computable in polynomial time, such that the committee $\vec{p}$ given by $f^*$ elementwise dominates $\vec{q}$. 
\end{lemma}
\begin{proof}
	Let $f$ be a feasible flow on network representation $\network$ and $\mathcal{R}_f$ be the residual network resulting from $f$. Also let $\vec{q}$ be the fractional committee given by $f$.
	If there is no augmenting path on $\mathcal{R}_f$, then $f$ is a max flow and the statement holds since $\vec{q}$ elementwise dominates $\vec{q}$.
	Otherwise, if there is some augmenting path, send flow along this path, and denote the resulting flow $f'$.
	The final arc in the augmenting path must be of the form $(c',t)$ for some $c'\in C$ by the structure of $\network$.
	Since $(c',t)$ is an arc in $\network$, it is clear that $f'(c', t)> f(c', t).$
	Note that the augmenting path can only contain one arc entering the sink (since otherwise it would contain a cycle), and thus $f'(c, t) = f(c, t)$ for all $c\in C\setminus \{c'\}.$
	It is now apparent that the fractional committee given by $f'$ elementwise dominates $\vec{q}$. 
	Thus, this property can be maintained while eliminating augmenting paths, which will lead to a max flow in polynomial time. 	
\end{proof}

The results of this section make use of a property which \citet{BrPe24a} have termed \emph{affordability}, which weakens the definition of priceability from \citet{PeSk20a} by removing the final condition that no unselected candidate can be afforded by its supporters.
\begin{definition}[Affordability]
	We say $W$ is \emph{affordable} if there exists a payment function $\pi_i: C, \mathbb{R}_{\geq 0}$ satisfying the following four conditions:
	\begin{enumerate}
		\item \label{enum:c1}
		$\pi_i(c)=0$ for each $i\in N$, $c\notin A_i$ 
		\item \label{enum:c2}
		$\sum_{c\in C} \pi_i(c) \leq k/n$ for each $i\in N$ 
		\item \label{enum:c3}
		$\sum_{i\in N} \pi_i(c) = 1$ for each $c\in W$ 
		\item \label{enum:c4}
		$\sum_{i\in N} \pi_i(c) = 0$ for each $c\notin W$. 
	\end{enumerate}
\end{definition}

The main theorem of this section shows that, for every affordable committee $W$, there is an ex-ante \GRP lottery for which every committee in its support contains $W$. 

\begin{theorem} \label{thm:gr_wp_bbw}
	Let $W$ be an affordable committee. There exists a lottery $\Delta = \{(\lambda_j, W_j)\}_{j \in [s]}$ such that (1) $\Delta$ satisfies \GRP and (2) $W\subseteq W_j$ for all $j\in [s]$. 
	Furthermore, given $W$, such a lottery can be computed in polynomial time.
\end{theorem}
\begin{proof}
	For some instance of our problem, let $\network$ be the network formulation and $W$ be some affordable committee.
	
	Now consider the flow $f$ on $\network$ defined as follows:
	\begin{itemize}
		\item $f(s,i) = \sum_{c\in C} \pi_i(c)$ for all $i\in N$
		\item $f(i,c) = \pi_i(c)$ for all $i\in N, c\in C$
		\item $f(c,t) = \sum_{i\in N} \pi_i(c)$ for all $c\in C$.
	\end{itemize}
	We claim that $f$ is a feasible flow on $\network$. 
	It is apparent that $f$ respects conservation of flows as each voter's incoming flow is equal to their outgoing flow and the same is true for candidate nodes.
	To see that the capacity constraints of $\network$ are respected by $f$, note that Condition \ref{enum:c2} ensures that $f(s,i)\leq k/n = \capacity(s,i)$ for all $i\in N$, Condition \ref{enum:c1} ensures that $f(i,c)=0=\capacity(i,c)$ for all $i\in N, c\notin A_i$, and Conditions \ref{enum:c3} and \ref{enum:c4} ensure that $f(c,t) \leq 1 = \capacity(c,t)$ for all $c\in C$. 

	Let $\vec{q}$ be the committee given by $f$. Since $f$ is feasible, we can conclude by \Cref{lem:mf_domination} that there exists a max flow $f^*$ such that the committee $\vec{p}$ given by $f^*$ elementwise dominates $\vec{q}$.
	Consider some implementation $\Delta$ of $\vec{p}$. 
	Since $p_c \geq q_c = f(c, t) = 1$ for all $c\in W$ by Condition \ref{enum:c3}, we have that every committee in the support of $\Delta$ contains $W$.  
	Lastly, because $\vec{p}$ is the fractional committee given by the max flow $f^*$ on $\network$, we have by \Cref{thm:gr_characterization} that $\vec{p}$ and thus $\Delta$ satisfies \GRP. 
\end{proof}

\citet{PSP21a} gave an exponential-time algorithm for an FJR committee and showed (\citep[Lemma 2]{PSP21a}) that this committee satisfies affordability.
This leads us to the following corollary to \Cref{thm:gr_wp_bbw}, which resolves an open question posed by \citet{ALS+23}.
\begin{corollary}
	A lottery satisfying ex-post FJR and ex-ante \GRP is guaranteed to exist.
\end{corollary}

\citet{ALS+23} also gave an algorithm using the \emph{Method of Equal Shares} (MES) \cite{PeSk20a} as a subroutine and showed that this algorithm gives EJR+, Strong UFS, and GFS. 
Since MES always returns a priceable (and thus affordable) committee, another consequence of \Cref{thm:gr_wp_bbw} is a strengthening of Theorem 4.1 from \citet{ALS+23}.
By applying the procedure described in the proof of \Cref{lem:mf_domination} to the result of MES, we can obtain a randomized committee which obtains all of the ex-post properties of MES in tandem with \GRP.\footnote{Note that rules belonging to BW-MES, the family of rules which give ex-post EJR+, ex-ante GFS, and ex-ante Strong UFS (Theorem 4.1, \citet{ALS+23}), do not necessarily satisfy \GRP. To see this, suppose one voter with budget left after the MES phase does not approve of any unselected candidate, while another voter has no budget remaining and approves some unselected candidate. An arbitrary allocation of the remaining budget need not satisfy \GRP with respect to two such voters.}
\begin{corollary}\label{crly:mes}
	A lottery satisfying ex-post EJR+ and ex-ante \GRP can be computed in polynomial time.
\end{corollary}

As these corollaries show,  \Cref{thm:gr_wp_bbw}  provides a useful tool that could aid in producing further best-of-both-worlds results in public decision problems.

\section*{Discussion}
In this paper, we introduced a new fairness axiom for probabilistic committee voting. 
We characterized our axiom and gave several algorithmic results exploring compatibility with efficiency, strategyproofness, and ex-post fairness.
Our characterization demonstrated a connection between fair committee voting and network flows. We believe that this network-based approach will prove useful for extensions of committee voting such as participatory budgeting.   

We showed that strategyproofness is incompatible with GRP-efficiency and GFS-efficiency (\Cref{prop:impossibility}).
We studied a weaker form of strategyproofness, known as excludable strategyproofness, and gave an algorithm which satisfies this property in addition to GRP-efficiency.
In the face of strong impossibilities involving strategyproofness in our own and related settings, we believe further study of excludable strategyproofness is a valuable research direction.
For example, when $k=1$, a strong impossibility result renders positive share, efficiency, and strategyproofness incompatible \citep{BBP+21a}.
However, to our knowledge, it is not known whether a rule which satisfies GFS, efficiency, and excludable strategyproofness exists, even for the special case when $k=1$.\footnote{Note that excludable strategyproofness can be obtained in conjunction with GFS and efficiency in isolation via CUT and Egalitarian welfare maximization, respectively \citep{ABM19a}.  }

\Cref{prop:impossibility} puts fair, efficient, and strategyproof rules even further from reach in the probabilistic committee voting setting than in the single-winner setting.
However, plenty of questions remain surrounding the compatibility of strategyproofness and fairness.
As we pointed out, random dictator satisfies strategyproofness and GFS.
However, to our knowledge, there is no known strategyproof rule that satisfies Strong UFS.
Since GRP implies Strong UFS, one would need to tackle this issue in order to prove compatibility of GRP and strategyproofness.
Our definition of strategyproofness can also be rewritten for rules which return lotteries, in which case it captures \emph{ex-ante} strategyproofness. 
We believe ex-ante strategyproofness is among the most intriguing directions for future work in best-of-both-worlds committee voting, particularly because \emph{ex-post} strategyproofness is incompatible with even a very weak form of ex-post fairness \citep{Pete21a}.

\section*{Acknowledgments}
	Mashbat Suzuki is supported by the ARC Laureate Project FL200100204 on ``Trustworthy AI''.

\clearpage
\bibliographystyle{ACM-Reference-Format}
\bibliography{abb,vollen,bibliography}

\clearpage
\appendix

\section*{Proof of Proposition~\ref{prop:grp_implies_pjr}} \label{app:grp_implies_pjr}
\begin{definition*}[PJR \citep{SEL+17a}]
	For any positive integer~$\ell$, a group of voters $N' \subseteq N$ is said to be \emph{$\ell$-cohesive} if $|N'| \geq \ell \cdot n/k$ and $\left| \bigcap_{i \in N'} \approvalBallotOfAgent{i} \right| \geq \ell$.
	A committee~$W$ is said to satisfy \emph{proportional justified representation (PJR)} if for every positive integer~$\ell$ and every $\ell$-cohesive group of voters~$N' \subseteq N$, it holds that $\left| \left( \bigcup_{i \in N'} \approvalBallotOfAgent{i} \right) \cap W \right| \geq \ell$.
\end{definition*}

\propPJR*
\begin{proof}
	Suppose, for some integral committee $W$, $\vec{p}=\vec{1}_W$ satisfies GRP.
	Let $N'\subseteq N$ be some $\ell$-cohesive group for some positive integer $\ell.$
	Since $\vec{p}$ satisfies GRP, we have
	\begin{align*}
		\left\vert W\cap \bigcup_{i\in N'} A_i \right\vert \ = \ \sum_{c\in \bigcup_{i\in N'} A_i} p_c \ &\geq  \ |N'| \frac{k}n - \max_{T\subseteq N'} \left[ |T|\frac{k}n - |\bigcup_{i\in T}A_i | \right] \\
		&\geq |N'| \frac{k}n - \max_{T\subseteq N'} \left[ |T|\frac{k}n - \ell \right] \ = \ \ell
	\end{align*}
	where the last inequality follows because $N'$ is $\ell$-cohesiveness and thus $|\cup_{i\in T} A_i| \geq |\cap_{i\in T} A_i| \geq \ell$ for all $T\subseteq N'$.
\end{proof}

\end{document}